\documentclass[acmsmall,screen,authorversion,nonacm]{acmart}
\pdfoutput=1
\usepackage[utf8]{inputenc} 
\usepackage{amsmath, amsthm}
\usepackage{thmtools}
\usepackage[english]{babel}
\usepackage{graphicx}
\usepackage{parskip}
\usepackage{array}
\usepackage{tabularx}
\usepackage{tikz} 
\usepackage{lineno}
\usepackage{stackengine}
\usepackage{diagbox}
\usepackage{multirow}
\usepackage{mdframed}
\usepackage{listings}
\usepackage{subcaption}
\usepackage{paralist}

\usepackage{pgfplots}
\usepackage[ruled,linesnumbered]{algorithm2e}
\usepackage[capitalise]{cleveref}
\usepackage{todonotes}
\usepackage{booktabs}
\usepackage{pgfplotstable}
\usepackage{colortbl}
\usepackage[group-separator={,}, group-minimum-digits=4, text-series-to-math, propagate-math-font]{siunitx}
\usetikzlibrary {arrows.meta}
\usetikzlibrary{arrows,positioning} 
\tikzset{
>=stealth',
punkt/.style={
       rectangle,
       rounded corners,
       draw=black, very thick,
       text width=6.5em,
       minimum height=2em,
       text centered},
pil/.style={
       ->,
       thick,
       shorten <=2pt,
       shorten >=2pt,},
main/.style={
    draw, rectangle,inner sep=4pt,
    }
}

\usetikzlibrary{automata,positioning,arrows}

\graphicspath{ {./} }

\definecolor{mGreen}{rgb}{0,0.6,0}
\definecolor{mGray}{rgb}{0.5,0.5,0.5}
\definecolor{mPurple}{rgb}{0.58,0,0.82}
\definecolor{backgroundColour}{rgb}{0.95,0.95,0.92}
\definecolor{mybluecolor}{rgb}{0.2549019607843137, 0.2117647058823529, 0.9823529411764706}
\def \darkred {black!20!red}

\lstdefinestyle{CStyle}{
tabsize = 2, 
showstringspaces = false, 
backgroundcolor=\color{backgroundColour},   
numbers = left, 
commentstyle = \color{green}, 
keywordstyle = \color{blue}, 
stringstyle = \color{red}, 
rulecolor = \color{black}, 
basicstyle = \small \ttfamily , 
breaklines = true, 
numberstyle = \tiny,
language=Java
}

\SetAlFnt{\normalsize}
\SetInd{0.4em}{0.4em}

\newcommand{\Grammar}{\mathcal{G}}
\newcommand{\Language}{\mathcal{L}}
\newcommand{\MCFG}{\operatorname{MCFG}}
\newcommand{\drMCFG}[2]{#1\text{-}\MCFG(#2)}

\newcommand{\MCFL}{\operatorname{MCFL}}
\newcommand{\drMCFL}[2]{#1\text{-}\MCFL(#2)}

\newcommand{\drMCFGs}[2]{#1\text{-}\operatorname{MCFGs}(#2)}

\newcommand{\NonTerminals}{\mathcal{N}}
\newcommand{\Terminals}{\Sigma}
\newcommand{\Rules}{\mathcal{R}}
\newcommand{\StartSymbol}{\mathcal{S}}
\newcommand{\Derives}{\Rightarrow}
\newcommand{\PathDerives}{\models}
\newcommand{\Endpoints}{\operatorname{Ends}}
\newcommand{\Label}{\lambda}

\newcommand{\TriangleL}{\mathcal{T}}
\newcommand{\Poly}{\operatorname{poly}}
\newcommand\numberthis{\addtocounter{equation}{1}\tag{\theequation}}
\newcommand{\Worklist}{\operatorname{W}}
\newcommand{\Done}{\operatorname{D}}
\newcommand{\ReplaceBy}{\mapsto}

\newcommand{\IDL}{\mathcal{I}_{\alpha,\beta}}

\newcommand{\Dyck}{\mathcal{D}}

\newcommand{\Paragraph}[1]{{\noindent \bf #1}}
\newcommand{\SubParagraph}[1]{{\noindent \em #1}}

\newcommand{\Path}{\rightsquigarrow}

\newcommand{\eps}{\varepsilon}
\newcommand{\Nats}{\mathbb{N}}

\newcommand{\op}{(}
\newcommand{\cp}{)}
\newcommand{\ob}{[}
\newcommand{\cb}{]}

\newcommand{\Project}{{\downharpoonright}}

\newcommand{\NumVectors}{m}
\newcommand{\NumVectorSets}{k}
\newcommand{\NumBits}{b}
\newcommand{\TemplateLanguage}{\mathcal{T}}

\newcommand{\pcOne}[1]{\textcolor{mybluecolor}{#1}}
\newcommand{\pcTwo}[1]{\textcolor{\darkred}{#1}}

\newcommand{\LTo}[1]{\xrightarrow{#1}}
\newcommand{\Delimiter}{|}

\title{Program Analysis via  Multiple Context Free Language Reachability}

\author{Giovanna Kobus Conrado}
\orcid{0000-0001-9474-6505}
\email{gkc@connect.ust.hk}
\affiliation{%
  \institution{Hong Kong University of Science and Technology}
  \country{Hong Kong}
}

\author{Adam Husted Kjelstr\o{}m}
\orcid{0000-0002-1275-9933}
\email{husted@cs.au.dk}
\affiliation{%
  \institution{Aarhus University}
  \country{Denmark}
}

\author{Andreas Pavlogiannis}
\orcid{0000-0002-8943-0722}
\email{pavlogiannis@cs.au.dk}
\affiliation{%
  \institution{Aarhus University}
  \country{Denmark}
}
\author{Jaco van de Pol}
\email{jaco@cs.au.dk}
\orcid{0000-0003-4305-0625}
\affiliation{%
  \institution{Aarhus University}
  \country{Denmark}
}

\makeatletter
\newcommand\stringCompare[2]{%
\pdf@strcmp{#1}{#2} %
}
\makeatother

\begin{document}

\begin{abstract}
\emph{Context-free language (CFL)} reachability is a standard approach in static analyses,
where the analysis question (e.g., is there a dataflow from $x$ to $y$?)
is phrased as a language reachability problem on a graph $G$ wrt a CFL $\mathcal{L}$.
However, CFLs lack the expressiveness needed for high analysis precision.
On the other hand, common formalisms for \emph{context-sensitive} languages are too expressive, in the sense that the corresponding reachability problem becomes undecidable.
\emph{Are there useful context-sensitive language-reachability models for static analysis?}

In this paper, we introduce \emph{Multiple Context-Free Language (MCFL)} reachability as an \emph{expressive} yet  \emph{tractable} model for static program analysis. 
MCFLs form an infinite hierarchy of mildly context sensitive languages parameterized by a \emph{dimension} $d$ and a \emph{rank} $r$.
Larger $d$ and $r$ yield progressively more expressive MCFLs, offering \emph{tunable} analysis precision.
We showcase the utility of MCFL reachability by developing a family of MCFLs that approximate interleaved Dyck reachability, a common but undecidable static analysis problem.

Given the increased expressiveness of MCFLs, one natural question pertains to their algorithmic complexity, i.e., \emph{how fast can MCFL reachability be computed?}
We show that the problem takes $O(n^{2d+1})$ time on a graph of $n$ nodes when $r=1$, and $O(n^{d(r+1)})$ time when $r>1$.
Moreover, we show that when $r=1$, even the simpler membership problem has a lower bound of $n^{2d}$ based on the Strong Exponential Time Hypothesis, while reachability for $d=1$ has a lower bound of $n^{3}$ based on the combinatorial Boolean Matrix Multiplication Hypothesis.
Thus, for $r=1$, our algorithm is optimal within a factor $n$ for all levels of the hierarchy based on the dimension $d$ (and fully optimal for $d=1$).

We implement our MCFL reachability algorithm and evaluate it by underapproximating interleaved Dyck reachability for a standard taint analysis for Android.
When combined with existing overapproximate methods, MCFL reachability discovers \emph{all} tainted information on 8 out of 11 benchmarks, while it has remarkable coverage (confirming $94.3\%$ of the reachable pairs reported by the overapproximation) on the remaining 3.
To our knowledge, this is the first report of high and provable coverage for this challenging benchmark set.
\end{abstract}



\maketitle

\section{Introduction}\label{SEC:INTRO}

\Paragraph{Static analysis via language reachability.}
Static analyses are a standard approach to determining program correctness, as well as other useful properties of programs.
They normally operate by establishing an approximate model for the program, effectively reducing questions about program behavior to algorithmic questions about the model.
One popular type of modeling in this direction is language reachability, where the program abstraction is via an edge-labeled graph $G$~\cite{Reps97,Reps1995}.
Language reachability is a generalization of standard graph reachability, parameterized by a language $\Language$.
Intuitively, given two nodes $u$, $v$, the problem asks to determine whether $v$ is reachable from $u$ via a path whose sequence of labels produce a string that belongs to $\Language$.
Such a path captures program executions that relate the objects represented by $v$ and $u$ (e.g., control-flow between two program locations or data-flow between two variables).
The role of $\Language$ is normally to increase the analysis precision by filtering out paths that represent spurious program executions, owing to the approximate nature of $G$.


\Paragraph{Context-free language reachability.}
Language-reachability based static analyses are most frequently phrased with respect to context free languages (CFL), known as CFL reachability, which have various uses.
For example, they are used to increase the precision of interprocedural analyses (aka whole-program analyses), where modeled executions cross function boundaries, and the analysis has to be \emph{calling-context sensitive}\footnote{It might sound paradoxical that a \emph{context-free} language makes the analysis \emph{context-sensitive}, but this is just a naming coincidence. ``Context sensitivity'' refers to calling contexts, and the CFL simulates the call stack.}.
A CFL $\Language$ captures that a path following an invocation from a caller function $\mathtt{foo}()$ to a callee $\mathtt{tie}()$ must return to the call site of $\mathtt{foo}()$ when $\mathtt{tie}()$ returns (as opposed to some other function $\mathsf{bar}()$ that also calls $\mathtt{tie}()$).
This approach is followed in a wide range of interprocedural static analyses, including data-flow and shape analysis~\cite{Reps1995}, type-based flow analysis~\cite{Rehof01} and taint analysis~\cite{Huang2015}, to name a few.
In practice, widely-used tools, such as Wala~\cite{Wala} and Soot~\cite{Bodden12}, equip CFL-reachability techniques to perform the analysis.

Another common use of CFLs is to track dataflow information between composite objects in a \emph{field-sensitive} manner.
Here, a CFL $\Language$ captures a dataflow between variables $x$ and $y$ if, for example $x$ flows into $z.f$ (i.e., field $f$ of composite object $z$), and $z.f$ itself flows into $y$ (as opposed to $z.g$ flowing into $y$).
This approach is standard in a plethora of pointer, alias and data dependence analyses~\cite{Lhotak2006,Reps97,Chatterjee2018,Sridharan2005,Sridharan2006,Lu2019}.

\Paragraph{The need for context-sensitive models.}
Although CFLs offer increased analysis precision over simpler, regular abstractions, there is often a need for more precise, \emph{context sensitive} models.
For example, almost all analyses are designed with both context and field sensitivity, as this leads to obvious precision improvements~\cite{Milanova2020,Lhotak2006,Sridharan2006,Spath2019}.
The underlying language modeling both sensitivities is the free interleaving  of the two corresponding CFLs, which is not a CFL, and is commonly phrased as \emph{interleaved Dyck reachability}.
Unfortunately, interleaved Dyck reachability is well-known to be undecidable~\cite{Reps2000}.
Since it is desirable to maintain at least some precision of each type, there have been various approximations that generally fall into two categories:
(i)~apply some kind of $k$-limiting, which approximates one of the CFLs by a regular (or even finite) language (the most common approach), or
(ii)~solve for each type of sensitivity independently~\cite{Spath2019}, possibly in a refinement loop~\cite{Ding2023,Conrado2024}.
Observe that both cases essentially fall back to context-free models, forgoing the desirable precision of context sensitivity for which the analysis was designed in the first place.
This leads to a natural question:

\emph{Are there natural, efficient (polynomial-time), and practically precise, context-sensitive approximations for interleaved Dyck reachability?}

\Paragraph{Multiple Context Free Languages.}
One of the most natural generalizations of CFLs towards mild context sensitivity is that of Multiple Context Free Languages (MCFLs)~\cite{Seki1991}.
These languages are generated by the corresponding Multiple Context Free Grammars (MCFGs), which form a hierarchy of expressiveness parameterized by a dimension $d$ and rank $r$ (concisely denoted as $\drMCFG{d}{r}$).
Intuitively, an MCFG in $d$ dimensions performs simultaneous context-free parsing on $d$ substrings of a word, and can thus capture bounded context-sensitivity between these substrings. 
The rank $r$ limits the number of non-terminals that can appear in a single production rule.
MCFGs have received considerable attention, as they are regarded as a realistic formalism for natural languages~\cite{Clark2015},
while several popular classes of formal languages fall into specific levels in this hierarchy,
e.g., CFLs are MCFLs of dimension $1$, and Tree Adjoining Languages (TALs)~\cite{Joshi1987} and Head Languages~\cite{Pollard1984} fall in dimension $2$~\cite{Seki1991}.
Despite the context sensitivity, for each $r\geq 2$, $\drMCFL{d}{r}$ forms a full abstract family of languages (AFL -- closed under homomorphism, inverse homomorphism, intersection with regular sets, union, and Kleene closure)~\cite{Rambow1999}, with decidable membership and emptiness~\cite{Vijay1987}.
As such, they form an elegant class of mildly context sensitive languages that are amenable to algorithmic treatment.

In a static analysis setting, language reachability with MCFLs has the potential to yield higher modeling power.
Moreover, this power is utilized in a \emph{controllable way}, owing to the higher expressivity along the MCFL hierarchy.
However, neither
(i)~the modeling power of MCFL reachability (\emph{what can MCFLs express in a program analysis setting?}) nor
(ii)~the algorithmic question (\emph{how fast can we solve MCFL reachability?}) have been studied.
This paper addresses these questions, by
\begin{compactenum}
\item designing a family of MCFLs for approximating the common static analysis problem of interleaved Dyck reachability, with remarkable coverage in practice, and
\item developing a generic algorithm for MCFL reachability (for any dimension $d$ and rank $r$), as well as proving fine-grained complexity lower bounds for the problem.
\end{compactenum}

The following motivating example illustrates the problem setting and our approach.

\subsection{Motivating Example}\label{SUBSEC:MOTIVATING_EXAMPLE}

In a standard dataflow analysis setting, the task is to identify pairs of variables $x$, $y$ such that the value of $x$ may affect the value of $y$.
This is achieved by following def-use chains in the program $P$.
The program model is a dataflow graph $G$, where nodes represent variables, and an edge $x\to y$ results from an instruction of the form $x=f(y)$ (for some uninterpreted function $f$).
To address the common issue of high false positives, the analysis must be both context-sensitive and field-sensitive.

For example, consider the program $P$ in \cref{fig:example_context_field}, and the dataflow graph $G$ in \cref{subfig:program_graph}.

\begin{figure}[!h] 
\begin{mdframed}[backgroundcolor=black!7!white,rightline=false,leftline=false,linewidth=0.25mm,]
\begin{subfigure}[t]{0.32\textwidth}
\centering
\begin{lstlisting}[numbers=left, language=C++,
tabsize=1,
basicstyle=\small,
numbersep=3pt,
numberstyle=\footnotesize\color{red},
showlines=true,
morekeywords={pair},
basicstyle=\ttfamily\small
]
pair tie(int x,int y){
  pair p;
  p.first = x;
  p.second = y;
  return p;
}
\end{lstlisting}
\end{subfigure}
\hfill
\begin{subfigure}[t]{0.3\textwidth}
\centering
\begin{lstlisting}[numbers=left, language=C++,basicstyle=\small,numbersep=3pt,numberstyle=\footnotesize\color{red},firstnumber=last,morekeywords={pair}]
void foo() {
  int a = 2;
  int b = 3;
  pair q = tie(a,b);
  int c = q.first;
  return;
}
\end{lstlisting}
\end{subfigure}
\hfill
\begin{subfigure}[t]{0.3\textwidth}
\centering
\begin{lstlisting}[numbers=left, language=C++,basicstyle=\small,numbersep=3pt,numberstyle=\footnotesize\color{red},showlines=true,firstnumber=last,morekeywords={pair}]
void bar() {
  int d = 5;
  pair r = tie(d,7);
  int e = r.second;
  return;
}
\end{lstlisting}
\end{subfigure}
\end{mdframed}
\caption{
A program $P$ containing three functions \texttt{tie()}, \texttt{foo()}, \texttt{bar()} and composite objects of type \texttt{pair}.
}
\label{fig:example_context_field}
\end{figure}

\begin{figure}[!h] 
\begin{subfigure}[t]{0.45\textwidth}
\centering
\begin{tikzpicture}[ line width=1pt, main/.style = {draw, rectangle,inner sep=4pt},scale=0.8] 
\node[main] (a) at (0,0){\footnotesize a}; 
\node[main] (d) at (0,-1){\footnotesize d}; 
\node[main] (b) at (0,-2){\footnotesize b}; 
\node[main] (x) at (1.4,-1){\footnotesize x}; 
\node[main] (y) at (1.4,-2){\footnotesize y}; 
\node[main] (p) at (2.4,-1){\footnotesize p};
\node[main] (ret) at (3.6,-1){\footnotesize $\text{tie}_{\text{ret}}$}; 
\node[main] (q) at (4.4,0){\footnotesize q}; 
\node[main] (r) at (4.4,-2){\footnotesize r}; 
\node[main] (c) at (5.4,0){\footnotesize c}; 
\node[main] (e) at (5.4,-2){\footnotesize e};

\draw [->] (a) -- (x) node [midway, above=2pt] {\footnotesize $\pcOne{(_{10}}$}; 
\draw [->] (d) -- (x) node [midway, below] {\footnotesize$\pcOne{(_{16}}$}; 
\draw [->] (b) -- (y) node [midway, below] {\footnotesize$\pcOne{(_{10}}$}; 
\draw [->] (x) -- (p) node [midway, above] {\footnotesize$\pcTwo{[_{1}}$}; 
\draw [->] (y) -- (p) node [midway, below=2pt] {\footnotesize$\pcTwo{[_{2}}$}; 
\draw [->] (p) -- (ret); 
\draw [->] (ret) -- (q) node [near end, left =2pt] {\footnotesize$\pcOne{)_{10}}$}; 
\draw [->] (ret) -- (r) node [near end, left=2pt] {\footnotesize$\pcOne{)_{16}}$}; 
\draw [->] (q) -- (c) node [midway, below] {\footnotesize$\pcTwo{]_{1}}$}; 
\draw [->] (r) -- (e) node [midway, above] {\footnotesize$\pcTwo{]_{2}}$};

\end{tikzpicture}
\subcaption{\label{subfig:program_graph}
A graph $G$ modeling context-sensitive and field-sensitive data flow in program $P$ from \cref{fig:example_context_field}.
}
\end{subfigure}
\qquad
\begin{subfigure}[t]{.45\linewidth}
\begin{tikzpicture}[ line width=1pt, main/.style = {draw, rectangle,inner sep=4pt},scale=0.8]
            \node[main] (e) at (0,0){\footnotesize e}; 
            \node[main] (f) at (0,-1.5){\footnotesize f}; 
            \node[main] (g) at (1.5,0){\footnotesize g}; 
            \node[main] (h) at (1.5,-1.5){\footnotesize h}; 
            \node[main] (i) at (3,0){\footnotesize i}; 
            \node[main] (j) at (4.5,0){\footnotesize j};
            \node[main] (k) at (4.5,-1.5){\footnotesize k}; 

            \node[main,white] (a) at (-1.5,-2.2){};

            \draw [->] (e) -- (g) node [midway, above=2pt] {\footnotesize $\pcTwo{[_{1}}$}; 
            \draw [->] (g) -- (h) node [midway, right=2pt] {\footnotesize $\pcOne{(_{100}}$}; 
            \draw [->] (h) -- (f) node [midway, above=2pt] {\footnotesize $\pcTwo{]_{1}}$};
            \draw [->] (f) -- (e) node [midway, left=2pt] {\footnotesize $\pcTwo{[_{1}}$}; 
            \draw [->] (g) -- (i) node [midway, above=2pt] {\footnotesize $\pcTwo{]_{1}}$}; 
            \draw [->] (i) -- (j) node [midway, above=2pt] {\footnotesize $\pcOne{)_{100}}$}; 
            \draw [->] (j) -- (k) node [midway, left=2pt] {\footnotesize $\pcTwo{]_{1}}$}; 
            
\end{tikzpicture}
\subcaption{\label{subfig:real_graph}
A subgraph of the \texttt{uranai} benchmark in a taint analysis for Android.
}
\end{subfigure}
\caption{
Two graphs modeling context and field sensitivity through edge labels.
}
\label{fig:example_context_field_graph}
\end{figure}

\Paragraph{Context sensitivity.}
Let us momentarily ignore edge labels in $G$.
We have a path $b\Path e$, signifying a dataflow from $b$ to $e$.
This, however, does not correspond to a valid program execution:~the path goes through the call of function \texttt{tie()} from \texttt{foo()} (where $b$ is declared), but when \texttt{tie()} returns, the execution continues on \texttt{foo}, rather than \texttt{bar()} where $e$ is declared.
Call-context sensitivity is achieved by modeling call sites using parenthesis labels, and only considering reachability as witnessed by paths that produce a properly balanced parenthesis string.
Formally we require that the label of the path forms a string that belongs to the Dyck language over parentheses (which is a CFL).
Now, the path $b\Path e$ is invalid, since $\pcOne{(_{10}}$ (along the edge $b\LTo{\pcOne{(_{10}}}y$) does not match $\pcOne{)_{16}}$ (along the edge $\text{tie}_{\text{ret}}\LTo{\pcOne{)_{16}}}r$), thus the analysis avoids reporting this false positive.

\Paragraph{Field sensitivity.}
With parentheses modeling call-context sensitivity, consider the path $d\Path e$.
The parenthesis string along this path is $\pcOne{(_{16}}\pcOne{)_{16}}$, which is balanced, representing the fact that call contexts are respected.
However, in $P$ there is no dataflow from $d$ to $e$, this time due to unmatched fields:~$x$ is assigned to $\mathit{p.first}$, and although there is a dataflow from $p$ to $r$, $e$ gets assigned $\mathit{r.second}$.
Field-sensitivity is achieved by modeling object fields using (square) bracket labels, and only considering reachability as witnessed by paths that produce a properly balanced bracket string.
Formally we require that the label of the path forms a string that belongs to the Dyck language over brackets.
Now, the path $d\Path e$ is invalid, since $\pcTwo{[_{1}}$ (along the edge $x\LTo{\pcTwo{[_{1}}}p$) does not match $\pcTwo{]_{2}}$ (along the edge $r\LTo{\pcTwo{]_{2}}}e$), thus the analysis avoids reporting this false positive.

\Paragraph{Context and field sensitivity, simultaneously.}
To capture both context and field sensitivity, the analysis must decide reachability via paths that are well-balanced wrt both parentheses and brackets.
However, these two types of symbols can be interleaved in an arbitrary way.
For example, out of all $6$ possible source-sink pairs $\{a,d,b\}\times \{c,e\}$, the only real dataflow is from $a$ to $c$, witnessed by a path producing the string $\pcOne{(_{10}}\pcTwo{[_{1}}\pcOne{)_{10}}\pcTwo{]_{1}}$.
As the corresponding reachability problem is undecidable~\cite{Reps2000}, existing techniques focus on overapproximating interleaved Dyck reachability, mostly by some context-free model.
This implies that these analysis results may still contain false positives in terms of reachability in the dataflow graph.

\Paragraph{Illustration on a real benchmark.}
To further illustrate the challenge, consider the dataflow graph in \cref{subfig:real_graph}, which is a subgraph of a common taint analysis for Android~\cite{Huang2015}.
From an overapproximation standpoint, consider the potential reachability from $e$ to $j$. 
Notice that there are valid context-sensitive paths and valid field-sensitive paths $e\Path j$; these are, respectively
\[
e\LTo{\pcTwo{[_1}} g \LTo{\pcTwo{]_1}}i \LTo{\pcOne{)_{100}}}j \qquad \text{and}\qquad  e\LTo{\pcTwo{[_1}} g \LTo{\pcOne{(_{100}}}h \LTo{\pcTwo{]_1}}f \LTo{\pcTwo{[_1}}e \LTo{\pcTwo{[_1}}g \LTo{\pcTwo{]_1}}i \LTo{\pcOne{)_{100}}}j
\]
Because of the presence of both paths, an overapproximation algorithm may fail to conclude that $e$ \emph{does not reach} $j$ through a path that is simultaneously context and field-sensitive.
In fact, even newer overapproximation methods such as \cite{Ding2023} indeed report that $e$ reaches $j$, thereby producing a false positive.
From an underapproximation standpoint, consider the reachability from $e$ to $k$, witnessed by the path
\[
e\LTo{\pcTwo{[_1}} g \LTo{\pcOne{(_{100}}}h \LTo{\pcTwo{]_1}}f \LTo{\pcTwo{[_1}}e \LTo{\pcTwo{[_1}}g \LTo{\pcTwo{]_1}}i \LTo{\pcOne{)_{100}}}j \LTo{\pcTwo{]_1} }k
\]
Observe that the path is non-simple, as we have to traverse the cycle once to obtain a valid string.
Moreover, the string interleaves parentheses with brackets, which means that it cannot be captured in a single Dyck language involving both parentheses and brackets.
In this work we demonstrate that MCFLs are an effective and tractable context-sensitive language formalism for underapproximating interleaved Dyck reachability that yields good approximations for real-world benchmarks.

\subsection{Summary of Results}\label{SUBSEC:CONTRIBUTIONS}
To benefit readability, we summarize here the main results of the paper, referring to the following sections for details.
We relegate all proofs to the Appendix.

\Paragraph{1. MCFL reachability as a program model.}
We introduce MCFL reachability as an expressive, yet tractable, context-sensitive formalism for static analyses.
Parameterized by the dimension $d$ and rank $r$, $\drMCFL{d}{r}$ yields an infinite hierarchy of progressively more expressive models that become Turing-complete in the limit.

We illustrate the usefulness of MCFL reachability by using it to under-approximate the (generally, undecidable) problem of interleaved Dyck reachability, which is the standard formulation of a plethora of static analyses.
In particular, for each $d\geq 1$, we obtain a $\drMCFL{d}{2}$ that achieves increased precision as $d$ increases (i.e., it discovers more reachable pairs of nodes), and becomes complete (i.e., it discovers all reachable pairs) in the limit of $d\to \infty$.
Our MCFL formulation is, to our knowledge, the first non-trivial method that approximates the reachability set from below, thus having no false positives.

Although underapproximations are less common in static analyses, they have many uses, such as excluding false positives~\cite{UA1,UA2}, reporting concrete witnesses, acting as a tool for bug-finding~\cite{UA3,UA4} and performing “must” analyses~\cite{UA5, UA6, UA7}.
Our underapproximation, when paired with existing overapproximate methods, allows limiting the set of potentially false negatives dramatically, and even find a fully-precise answer (as often is the case in our experiments).

\Paragraph{2. MCFL reachability algorithm.}
We develop a generic algorithm for solving $\drMCFL{d}{r}$ reachability, for any value of $d$ and $r$.
Our algorithm generalizes the existing algorithms for CFL reachability~\cite{Yannakakis1990} and TAL reachability~\cite{Tang2017}.
In particular, we establish the following theorem.

\begin{restatable}{theorem}{thmupperbound}\label{thm:upper_bound}
All-pairs $\drMCFL{d}{r}$-reachability given a grammar $\Grammar$ on a graph $G$ of $n$ nodes can be solved in
\begin{compactenum}
\item\label{item:upper_boundr1} $O(\Poly(|\Grammar|)\cdot \delta\cdot n^{2d})$ time, if $r=1$, where $\delta$ is the maximum degree of $G$, and
\item\label{item:upper_boundrlarge} $O(\Poly(|\Grammar|)\cdot n^{d(r+1)})$ time, if $r>1$.
\end{compactenum}
\end{restatable}

As CFLs and TALs are $\drMCFL{1}{2}$ and $\drMCFL{2}{2}$, respectively, \cref{thm:upper_bound} recovers the known bounds of $O(n^3)$ and $O(n^6)$ for the corresponding reachability problems.
We also remark that the simpler problem of $\drMCFL{d}{r}$ membership is solved in time $O(n^{d(r+1)})$ time on strings of length $n$~\cite{Seki1991}.
\cref{thm:upper_bound} states that reachability is no harder than membership, as long as the current bounds hold, for bounded-degree graphs ($\delta=O(1)$) or when $r>1$.

\Paragraph{3. MCFL membership and reachability lower bounds.}
Observe that the bounds in \cref{thm:upper_bound} grow exponentially on the dimension $d$ and rank $r$ of the language.
The next natural question is whether this dependency is tight, or it can be improved further.
Given the role of MCFL reachability as an abstraction mechanism, this question is also practically relevant.
For example, consider a scenario where a 3-dimensional MCFL is used in a static analysis setting, but the analysis is too heavy for the task at hand.
The designer faces a dilemma:~\emph{``should we attempt to improve the analysis algorithm, or should we find a simpler model, e.g., based on a 2-dimensional MCFL?''}.
A proven lower bound resolves this dilemma in favor of receding to 2 dimensions\footnote{Of course, one should also look for heuristics that offer practical speedups. We touch on this in \cref{SEC:CONCLUSION}.}.
We prove two such lower bounds based on arguments from fine-grained complexity theory.

First, we study the dependency of the exponent on the dimension $d$.
For this, we fix $r=1$ and arbitrary $d$, for which the membership problem, as well as the reachability problem on bounded-degree graphs, takes $O(n^{2d})$ time.
We establish a lower-bound of $n^{2d}$ based on the Strong Exponential Time Hypothesis (SETH).

\begin{restatable}{theorem}{thmovhard}\label{thm:ov_hard}
For any integer $d$ and any fixed $\epsilon>0$,
the $\drMCFL{d}{1}$ membership problem on strings of length $n$
has no algorithm in time $O(n^{2d-\epsilon})$, under SETH.
\end{restatable}

\cref{thm:ov_hard} is based on a fine-grained reduction from Orthogonal Vectors.
The $\NumVectorSets$-Orthogonal Vectors (OV) problem asks, given a set of $\NumVectors\cdot \NumVectorSets$ Boolean vectors, to identify $\NumVectorSets$  vectors that are orthogonal.
The corresponding hypothesis $\NumVectorSets$-OVH states that this problem cannot be solved in $O(\NumVectors^{\NumVectorSets-\epsilon})$ time, for any fixed $\epsilon>0$ (it is also known that SETH implies $\NumVectorSets$-OVH~\cite{Williams05}).
\cref{thm:ov_hard} is obtained by proving that a $\drMCFL{d}{1}$ can express the orthogonality of $2d$ vectors.
This implies that the dependency $2d$ in the exponent of \cref{thm:upper_bound} cannot be improved, while for $r=1$ our reachability algorithm is optimal on sparse graphs.

Second, note that, on dense graphs (i.e., when $\delta=\Theta(n)$), the bound in  \cref{thm:upper_bound} \cref{item:upper_boundr1} is a factor $n$ worse than the lower bound of \cref{thm:ov_hard}.
Are further improvements possible in this case?
To address this question, we focus on the case of $d=1$, for which this upper bound becomes $O(n^3)$.
We show that the problem has no subcubic combinatorial algorithm based on the combinatorial Boolean Matrix Multiplication Hypothesis (BMMH).

\begin{restatable}{theorem}{thmtrianglehard}\label{thm:triangle_hard}
For any fixed $\epsilon>0$,
the single-pair $\drMCFL{1}{1}$-reachability problem on graphs of $n$ nodes
has no algorithm in time $O(n^{3-\epsilon})$ under BMMH.
\end{restatable}

Hence, the $\delta$ factor increase in the complexity cannot be improved in general,
while \cref{thm:upper_bound} is tight for $d=1$ and $r=1$, among combinatorial algorithms.

\Paragraph{4. Implementation and experimental evaluation.}
We implement our algorithm for MCFL reachability and run it with our family of $\drMCFL{d}{2}$s on standard benchmarks of interleaved Dyck reachability that capture taint analysis for Android~\cite{Huang2015}.
To get an indication of coverage, we compare our underapproximations with recent overapproximations.
Remarkably, our underapproximation matches the overapproximation on most benchmarks, meaning that we have fully sound and complete results.
For the remaining benchmarks, our underapproximation is able to confirm ($94.3\%$) of the taint information reported by the overapproximation.
To our knowledge, this is the first report of such high, provable coverage for this challenging benchmark set.

\section{Preliminaries}\label{SEC:PRELIMINARIES}

In this section, we introduce MCFLs and the problem of MCFL reachability.

\Paragraph{General notation.}
Given an integer $n$, we let $[n]=\{1,\dots, n\}$.
Given a finite set of symbols $\Terminals$, a string is a sequence $w\in \Terminals^*$.
We denote by $|w|$ the length of $w$,
and we write $w[a\ReplaceBy b]$ to denote the string that occurs after replacing every instance of $a$ with $b$.
Given two strings $w_1$ and $w_2$, we denote their concatenation by $w_1w_2$.
We use $\epsilon$ to denote the empty string.
As a warm-up, we begin with the standard definition of context-free grammars, in notation somewhat different than usual, which extends naturally to multiple dimensions afterwards.

\Paragraph{Context-free grammars.} A \emph{context-free grammar} is a tuple $\Grammar=(\NonTerminals, \Terminals, \Rules, \StartSymbol)$ such that the following hold.
\begin{compactitem}
\item $\NonTerminals$ is a finite alphabet of non-terminals, each denoted as a predicate of the form $A(x)$.
\item $\Terminals$ is a finite alphabet of terminals.
\item $\Rules$ is a finite set of rules.
Each rule is either a basic rule, meaning a predicate $A(s)$ with $s\in \Terminals^*$, or a production rule of the form
\[
A_0(s) \gets A_1(x^1) ,\cdots, A_{\ell}(x^{\ell})
\numberthis\label{eq:CFG_production_rule}
\]
such that
(i)~each $A_i$ is a non-terminal, 
(ii)~the variables $x^i$ are pairwise distinct,
(iii)~$s\in (\Terminals\cup X)^*$ where $X=\{x^i \mid i\in[\ell]\}$, and
(iv)~each $x^i$ appears exactly once in $s$.
\item $\StartSymbol$ is the initial non-terminal.
\end{compactitem}

\SubParagraph{Example.} Consider the language comprised of strings of the form $\{0^{n}1^{n}1^{m}0^{m}\vert n,m\geq 0\}$. A context-free grammar that identifies it is the tuple $\Grammar=(\NonTerminals, \Terminals, \Rules, \StartSymbol)$ such that:

\begin{compactitem}
\item $\NonTerminals=\{S,A,B\}$.
\item $\Terminals=\{0,1\}$.
\item $\Rules$ is comprised of the basic rules $A(\epsilon)$ and $B(\epsilon)$ and the production rules:
\begin{align}
&A(0x1) \gets A(x)\\
&B(1x0) \gets B(x)\\
&\StartSymbol(x^1x^2) \gets A(x^1),B(x^2)\label{rule:cfg-conc}
\end{align}
\end{compactitem}
 
Intuitively, $A$ parses strings of the type $\{0^{n}1^{n}\vert n\geq 0\}$ and $B$ parses strings of the type $\{1^{m}0^{m}\vert m\geq 0\}$. 
Rule \ref{rule:cfg-conc} can be read as follows: if there are 2 substrings $x^1,x^2$ such that $A$ parses $x^1$ (represented by $A(x^1)$ on the RHS) and $B$ parses $x^2$ (represented by $B(x^2)$ on the RHS), then $\StartSymbol$ parses the string formed by concatenating $x^1$ with $x^2$ (represented by $\StartSymbol(x^1x^2)$ on the LHS).

\Paragraph{Multiple context-free grammars.}
Context-free grammars can be naturally extended towards context sensitivity by allowing predicates to take more than one variable as a parameter. 
A \emph{multiple context-free grammar}~\cite{Seki1991} has dimension $d$ and a rank $r$, such that $d$ bounds the number of parameters taken by a non-terminal, and $r$ bounds the number of non-terminals on the right-hand side of any production rule. 
Formally, given integers $d$ and $r$, a multiple context-free grammar of dimension $d$ and rank $r$, $\drMCFG{d}{r}$, is a grammar $\Grammar=(\NonTerminals, \Terminals, \Rules, \StartSymbol)$ such that the following hold.
\begin{compactitem}
\item $\NonTerminals$ is a finite alphabet of non-terminals, each denoted as a predicate of the form $A(x_1, \dots, x_k)$, for some arity $k\leq d$.
\item $\Terminals$ is a finite alphabet of terminals.
\item $\Rules$ is a finite set of rules.
Each rule is either a basic rule, meaning a non-terminal $A(s_1, \dots, s_k)$ with $s_i\in \Terminals^*$, or a production rule of the form
\[
A_0(s_1,\dots, s_{k_0}) \gets A_1(x_1^1,\dots, x_{k_1}^1) ,\cdots, A_{\ell}(x_1^{\ell},\dots, x_{k_{\ell}}^{\ell})
\numberthis\label{eq:production_rule}
\]
such that
(i)~$\ell\leq r$ and each $A_i$ is a non-terminal of arity $k_i$, 
(ii)~the variables $x_j^i$ are pairwise distinct,
(iii)~each $s_i\in (\Terminals\cup X)^*$ where $X=\{x_{j}^i \mid i\in[\ell],j\in[k_i]\}$, and
(iv)~each variable has at most one appearance in $s_1\cdots s_{k_0}$, i.e., each $x^i_j$  appears in at most one $s_\ell$, and at most once in that $s_\ell$.
\item $\StartSymbol$ is the initial non-terminal, of arity $1$.
\end{compactitem}
The size of the grammar $|\Grammar|$ is the sum of the length of all of its rules. The length of a rule is the sum of total number of variables and the total length, in characters, of $s_1\cdots s_{k_0}$.

\Paragraph{Derivations and languages.}
The language of $\Grammar$ is defined inductively as a set of derivable string tuples.
In particular, for any non-terminal $A$ of arity $k$ and tuple of strings $(w_1,\dots, w_k)$, the derivation relation $A\Derives(w_1,\dots, w_k)$ is the smallest relation that satisfies the following.
First, for every basic rule $A(w_1,\dots, w_k)$, we have $A\Derives(w_1,\dots, w_k)$.
Second, for every production rule of the form of \cref{eq:production_rule}, if
$A_j\Derives(w_1^j,\dots, w_{k_j}^j)$ for every $j\in [\ell]$, we have $A_0\Derives(w_1^0,\dots, w_{k_0}^0)$, where
\[
w_i^0=s_i[x_{l}^{j}\ReplaceBy w_{l}^{j}]_{j\in [\ell], l\in[k_j]}.
\] 
Intuitively, a predicate represents the simultaneous derivation of a tuple of strings (as opposed to a single string in CFLs).
The language of $\Grammar$ is defined as $\Language(\Grammar)=\{ w\in \Terminals^*\colon \StartSymbol\Derives w \}$.
The membership question, i.e., given a word $w$ of length $n$, is it the case that $w\in \Language(\Grammar)$,
is solvable in $O(n^{d(r+1)})$ time~\cite{Seki1991}.
Observe that the class of CFLs is precisely the class $\drMCFL{1}{r}$ of 1-dimensional $\MCFL$s.
Despite being context-sensitive, MCFLs carry various desirable properties of CFLs, e.g., they form a fully abstract family of languages (AFL).

\SubParagraph{Example.}
Consider the mildly context-sensitive language $\Language=\{w_1w_2\#w_2w_1\mid w_1,w_2\in\{0,1\}^*\}$.
The $\drMCFG{2}{2}$ containing the following rules produces $\Language$.
See \cref{fig:mcfl_example} for an example.
\begin{align}
        &A(\epsilon,\epsilon)\label{rule:ex_mcfl_e}\\
        &A(x_10,x_20)\gets A(x_1,x_2)\label{rule:ex_mcfl_0}\\
        &A(x_11,x_21)\gets A(x_1,x_2)\label{rule:ex_mcfl_1}\\
        &\StartSymbol(x_1y_1\#y_2x_2)\gets A(x_1,x_2), A(y_1,y_2)\label{rule:ex_mcfl_f}
\end{align}

Rules \ref{rule:ex_mcfl_e}, \ref{rule:ex_mcfl_0} and \ref{rule:ex_mcfl_1} are such that $A$ parses the pair $(x_1,x_2)$ if and only if the strings $x_1$ and $x_2$ are equal. In other words, $A$ parses copies of some arbitrary binary string. Rule \ref{rule:ex_mcfl_f} can be intuitively read as follows: If there are 4 substrings $x_1,x_2,y_1,y_2$ such that $A$ parses the pair $(x_1,x_2)$ and $A$ also parses the pair $(y_1,y_2)$, then $S$ parses the string formed by concatenating $x_1$ with $y_1$, adding a $\#$ at the end, and further concatenating it with $y_2$ and $x_2$ (represented by $\StartSymbol(x_1y_1\#y_2x_2)$ on the LHS). Notice that, given the tuples recognized by $A$, there is a straightforward correspondence between $ x_1y_1\#y_2x_2$ in Rule~\ref{rule:ex_mcfl_f} and $w_1w_2\#w_2w_1$ in the language definition.

\begin{figure}
\centering
\begin{tikzpicture}[node distance={15mm}, line width=1pt, main/.style = {draw, rectangle},scale=0.8] 
    \node[main] at (-3,-0.5) (01) [white]{$\{1, 2\}$}; 
    \node[main] at (-3,-1.5) (02) [white]{$\{1, 2\}$}; 
    \node[main] at (0,-0.5) (ee1) [black, text=black]{\footnotesize$A(\epsilon,\epsilon)$}; 
    \node[main] at (0,-1.5) (ee2) [black, text=black]{\footnotesize$A(\epsilon,\epsilon)$}; 
    \node[main] at (3,-0.5) (001) [black, text=black]{\footnotesize$A(0,0)$}; 
    \node[main] at (3,-1.5) (002) [black, text=black]{\footnotesize$A(0,0)$};  
    \node[main] at (6,-0.5) (0101) [black, text=black]{\footnotesize$A(01,01)$}; 
    \node[main] at (6,-1.5) (0000) [black, text=black]{\footnotesize$A(00,00)$};  
    \node[main] at (10,-1) (01000001) [black, text=black]{\footnotesize$\StartSymbol(0100\#0001)$};  
    \draw [->] (01) -- (ee1)  node [midway,above] {\footnotesize\ref{rule:ex_mcfl_e}};
    \draw [->] (02) -- (ee2)  node [midway,above] {\footnotesize\ref{rule:ex_mcfl_e}};
    \draw [->] (ee1) -- (001)  node [midway,above] {\footnotesize\ref{rule:ex_mcfl_0}};
    \draw [->] (ee2) -- (002)  node [midway,above] {\footnotesize\ref{rule:ex_mcfl_0}};
    \draw [->] (001) -- (0101)  node [midway,above] {\footnotesize\ref{rule:ex_mcfl_1}};
    \draw [->] (002) -- (0000)  node [midway,above] {\footnotesize\ref{rule:ex_mcfl_0}};
    \draw [->] (0101) -- (7.5,-0.5) -- (7.5,-1) --  (01000001)  node [midway,above] {\footnotesize\ref{rule:ex_mcfl_f}};
    \draw [] (0000) -- (7.5,-1.5) -- (7.5,-1);
\end{tikzpicture}
\caption{
A derivation tree corresponding for the MCFG $\Grammar$ for the Language $\Language=\{w_1w_2\#w_2w_1\mid w_1,w_2\in\{0,1\}^*\}$, executed on the word $0100\#0001$. 
Edge labels indicate the corresponding derivation rule.
}
\label{fig:mcfl_example}
\end{figure}

\Paragraph{Special classes of MCFGs.}
The grammar $\Grammar$ is called \emph{non-deleting} if in every production of the form of \cref{eq:production_rule},
every variable $x_i^j$ appears exactly once in $s_1,\dots, s_{k_0}$.
Moreover, $\Grammar$ is called non-permuting if for every $i\in[\ell]$, the order of appearance of all variables $x^i_{1},\dots, x^i_{k_i}$ in $s_1\cdots s_k$ does not change.
It is known that, for every dimension $d$, the class MCFGs which are both non-deleting and non-permuting is as expressive as the whole class in that dimension~\cite{Kracht2003}.
Moreover, the following trade-off property holds for all $r\geq 2$: $\drMCFL{d}{r+k}\subseteq\drMCFL{(k+1)d}{r}$~\cite{Rambow1999},
and in particular, $\drMCFL{d}{r}\subseteq\drMCFL{(r-1)d}{2}$. 
In the specific case of CFLs, or $\drMCFL{1}{r}$, higher rank does not yield higher expressiveness for $r\geq 2$, i.e, $\drMCFL{1}{r}=\drMCFL{1}{r+1}$.
This, however, does not hold for $d\geq 2$.

\Paragraph{MCFL reachability.}
Given a set of terminal symbols $\Terminals$,
we consider labeled directed graphs $G=(V,E)$ where 
$V$ is a set of $n$ nodes, and
$E\subseteq V\times\Terminals\times V$ is a set of $m$ labeled edges.
We define the labeling function $\Label\colon E\to \Terminals$ by $\Label(u,a,v)=a$.
A path in $G$ is a finite sequence of connected edges $P=(e_1, \dots e_k)$.
We lift the labeling function to paths, and let $\Label(P)=\Label(e_1)\cdots \Label(e_k)$ be the string obtained by concatenating the labels of the edges of $P$.

Given a language $\Language \subseteq \Terminals^*$, we say that a node $v$ is $\Language$-reachable from another node $u$ if there is a path $P\colon u\Path v$ with $\Label(P)\in \Language$~\cite{Yannakakis1990}.
Thus language reachability refines the standard notion of graph reachability to witnesses that produce a string belonging to the language.
Language reachability is most frequently phrased wrt a CFL language, known as CFL-reachability, and
has found truly numerous applications in databases, program analyses, and others.
In this paper we consider the more expressive setting of $\drMCFL{d}{r}$-reachability, parameterized by the dimension $d$ and rank $r$ of the language.

\section{MCFL Approximations of Interleaved Dyck Reachability}\label{SEC:MODELING}

In this section we present a hierarchy of MCFLs of rank $2$, indexed by the dimension $d$, for underapproximating the interleaved Dyck language, i.e, every string recognized by the MCFL is a valid interleaved Dyck string, while more such strings are captured as the dimension increases.
This section is organized as follows.
In \cref{SUBSEC:MODELING_PRELIMS} we introduce formally the problem of interleaved Dyck reachability.
In \cref{SUBSEC:MODELING_MCFG_MODEL} we present the general structure of our MCFGs, namely the predicates they contain and the intuition behind them.
In \cref{SUBSEC:BASIC_RULES} we develop a basic set of production rules that achieve the desired approximation, but the MCFGs are relatively less expressive (i.e., some ``simple'' strings require a high dimension to be produced).
Finally, in \cref{SUBSEC:EXPRESSIVENESS} we introduce some additional, more involved rules, that make the MCFGs significantly more expressive.

\subsection{Preliminaries}\label{SUBSEC:MODELING_PRELIMS}

\Paragraph{General notation.} 
Given a language $\Language$, let $\Sigma(\Language)$ be the alphabet of $\Language$, i.e., the set of letters that occur in some string in $\Language$.
Given some alphabet $A$ and some string $s$ over a superset of $A$, we denote by $s\Project A$ the projection of $s$ to $A$, i.e., the subsequence of $s$ consisting of all the letters in $A$. 
Given two languages $\Language_1$ and $\Language_2$ with disjoint alphabets ($\Sigma(\Language_1)\cap \Sigma(\Language_2)=\emptyset$), their interleaving is defined as the set of strings $s$ such that when projected to the alphabet of each language, we obtain a string that belongs to that language.
Formally,
\[
\Language_1\odot \Language_2 =\{ s\in (\Sigma(\Language_1)\cup \Sigma(\Language_2) )^* \colon  (s\Project \Sigma(\Language_1)) \in \Language_1 \text{ and } (s\Project \Sigma(\Language_2)) \in \Language_2 \} 
\]

\Paragraph{Dyck languages.}  
Given a natural number $k\in \mathbb{N}$ and a set of parenthesis pairs $\{\pcOne{\op_i},\pcOne{\cp_i}\}_{i\in [k]}$, the Dyck language $\Dyck$ is the language of properly-balanced parentheses, produced by the context-free grammar $\Grammar=(\{\StartSymbol\}, \{ \pcOne{\op_i},\pcOne{\cp_i}\}_{i \in [k]}, \Rules, \StartSymbol)$, where $\Rules$ is the following set of production rules.
\begin{align}
&\StartSymbol(\epsilon)  \label{rule-epsilon}\\
&\StartSymbol(\pcOne{\op_i} x \pcOne{\cp_i}) \gets \StartSymbol(x) \label{rule-parentheses} \quad \forall i \in [k]\\        
&\StartSymbol(x_1x_2) \gets \StartSymbol(x_1),\StartSymbol(x_2) \label{rule-concatenation}
\end{align}

\Paragraph{Interleaved Dyck languages.} 
Given a natural number $k \in \mathbb{N}$, we define two Dyck languages, $\Dyck_\alpha$ and $\Dyck_{\beta}$,
over the alphabets $\{\{ \pcOne{\op_i},\pcOne{\cp_i}\}_{i \in [k]}\}$ and  $\{\{ \pcTwo{\ob_i},\pcTwo{\cb_i}\}_{i \in [k]}\}$, respectively. 
The interleaved Dyck language of $\Dyck_{\alpha}$ and $\Dyck_{\beta}$ is defined as  $\IDL := \Dyck_\alpha \odot \Dyck_\beta$.
In  words, $\IDL$ consists of strings that may contain both parentheses and brackets, as long as the substring consisting of its parentheses (resp., brackets) is in $\Dyck_\alpha$ (resp., $\Dyck_\beta$). 

\SubParagraph{Example.} For $k=2$, the string $s=\pcOne{(_1}\pcTwo{[_1}\pcOne{)_1}\pcTwo{[_2}\pcTwo{]_2}\pcOne{(_2}\pcTwo{]_1}\pcOne{)_2}$ is in $\IDL$ since $s\Project\Sigma(\Dyck_\alpha)=\pcOne{(_1}\pcOne{)_1}\pcOne{(_2}\pcOne{)_2}$ is in $\Dyck_\alpha$ and $s\Project\Sigma(\Dyck_\beta)=\pcTwo{[_1}\pcTwo{[_2}\pcTwo{]_2}\pcTwo{]_1}$ is in $\Dyck_\beta$.

As we have already described in \cref{SEC:INTRO} (and illustrated in \cref{SUBSEC:MOTIVATING_EXAMPLE}), language reachability wrt $\IDL$ is a frequent model in static analyses, but also undecidable~\cite{Reps2000}.
Next, we develop a family of MCFLs that approximate $\IDL$ with increasing precision.

\subsection{MCFG Model}\label{SUBSEC:MODELING_MCFG_MODEL}

We define a sequence of grammars $(\Grammar_d^+)_{d\geq 1}$ such that each $\Grammar_d^+$ is a $\drMCFG{d}{2}$ that underapproximates $\IDL$. 
In particular, we have 
(i)~$\Grammar_d^+\subseteq \IDL$, and
(ii)~for every string $s\in \IDL$ there is a $d_0\in \Nats$ such that for every $d\geq d_0$, we have $s\in \Language(\Grammar_d^+)$. 
In words, $(\Grammar_d^+)_{d\geq 1}$ underapproximates $\IDL$ and coincides with $\IDL$ in the limit $d\to \infty$.

Let $k$ be the number of different parenthesis/bracket pairs in $\IDL$ (for simplicity in notation, we assume they are equal).
Each $\Grammar_d^+$ consists of the following.
\begin{compactitem}
\item A set of nonterminals $\{\StartSymbol\} \cup \{P^{c}\}_{{c} \in [d]} \cup \{Q^{c}\}_{{c} \in [d]}$,
with each $P^c$ and $Q^c$ having arity $c$.
$\StartSymbol$ is the initial nonterminal.
\item A set of terminals $\{ \pcOne{\op_i},\pcOne{\cp_i}\}_{i \in [k]}\cup \{ \pcTwo{\ob_i},\pcTwo{\cb_i}\}_{i \in [k]}$, which are the parentheses and brackets of $\IDL$.
\item A set of production rules $\Rules$.
\end{compactitem}

It remains to define the set of rules $\Rules$.
To benefit readability, we first define in \cref{SUBSEC:BASIC_RULES} a simpler set $\Rules$, leading to an intermediate grammar $\Grammar_d^{\circ}$ that has the desired approximation properties (i) and (ii), but is relatively less expressive. 
Then we strengthen $\Rules$ in \cref{SUBSEC:EXPRESSIVENESS} to obtain $\Grammar_d^{+}$, which is more expressive (i.e., $\Language(\Grammar_d^{\circ})\subseteq \Language(\Grammar_d^{+})$). Intuitively, in $\Grammar_d^{\circ}$, $P^{c}$ seeks to capture strings of parentheses split into parts, and $Q^{c}$ strings of brackets split into parts. These ``parts'' can then be interleaved to form a string captured by $\StartSymbol$.  The grammar $\Grammar_d^{+}$ loosens these restrictions to increase expressivity.

The grammars presented model the interleaving of two Dyck languages, but this can be extended. This requires incorporating a set of nonterminals (analogous to $\{P^{c}\}_{{c} \in [d]}$ and $\{Q^{c}\}_{{c}\in [d]}$) and a set of terminals (analogous to $\{ \pcOne{\op_i},\pcOne{\cp_i}\}_{i \in [k]}$ and $\{ \pcTwo{\ob_i},\pcTwo{\cb_i}\}_{i \in [k]}$) for each additional language.

The rules for each set on nonterminals will be analogous to the ones on $\{P^{c}\}_{{c} \in [d]}$ and $\{Q^{c}\}_{{c}\in [d]}$, and the nonterminals corresponding to each language will be interleaved to generate $\StartSymbol$. This will lead to a final grammar of rank equal to the number of Dyck CFLs being interleaved.

\subsection{The Simpler Grammar $\Grammar_d^{\circ}$}\label{SUBSEC:BASIC_RULES}

We now present the initial set of production rules $\Rules$ towards the simpler grammar $\Grammar_d^{\circ}$.
The intuitive principle is as follows:~a string in $\IDL$ such as $s=\pcOne{(_1}\pcTwo{[_1}\pcOne{)_1}\pcTwo{[_2}\pcTwo{]_2}\pcOne{(_2}\pcTwo{]_1}\pcOne{)_2}$ can be partitioned into substrings $\{s_i\}_i$ so that $s_i\in \Sigma(\Dyck_\alpha)^*$ or $s_i\in \Sigma(\Dyck_\beta)^*$. 
For example, one such partitioning for $s$ is $\{\pcOne{(_1},\pcTwo{[_1},\pcOne{)_1},\pcTwo{[_2}\pcTwo{]_2},\pcOne{(_2},\pcTwo{]_1},\pcOne{)_2}\}$.
Notice that the substrings from $\Dyck_\alpha$, $\{\pcOne{(_1}, \pcOne{)_1},\pcOne{(_2},\pcOne{)_2}\}$, concatenated, form the string $\pcOne{(_1}\pcOne{)_1}\pcOne{(_2}\pcOne{)_2}\in \Dyck_\alpha$,
while substrings from $\Dyck_\beta$, $\{\pcTwo{[_1},\pcTwo{[_2}\pcTwo{]_2},\pcTwo{]_1}\}$ form the string $\pcTwo{[_1}\pcTwo{[_2}\pcTwo{]_2}\pcTwo{]_1}\in \Dyck_\beta$. 

Our MCFG $\Grammar_d^{\circ}$ emulates this process:~it uses the $d$ dimensions to proactively partition strings $s_{\alpha}\in \Dyck_{\alpha}$ and $s_{\beta}\in \Dyck_{\beta}$ into $d$ substrings,
which will be the interleaving points to produce a string $s\in \IDL$ derived by the initial nonterminal $\StartSymbol$.
This emulates a scheme of bounded context switching.
With higher dimensions $d$ we can afford more interleaving points, increasing the expressiveness of $\Grammar_d^{\circ}$.
In particular, for any $c\in[d]$, $P^{c}$ (resp., $Q^{c})$ is such that for every string $s\in \Dyck_{\alpha}$ (resp., $s\in \Dyck_{\beta})$ that can be partitioned into substrings $s=s_1\cdots s_c$, we have $P^{c}(s_1,\cdots, s_c)$ (resp., $Q^{c}(s_1,\cdots, s_c)$).

We start with rules for the predicates $P^c(x_1,\cdots, x_c)$, which recognize strings $x_1\cdots x_c\in \Dyck_{\alpha}$. The case of $c=1$ captures precisely $\Dyck_\alpha$.
\begin{align}
&P^1(\epsilon) \label{subs-rule-eps-1} \\
&P^1(\pcOne{\op_i} x \pcOne{\cp_i}) \gets P^1(x) \qquad \forall i \in [k] \label{subs-rule-par-1} \\     
&P^1(x_1x_2) \gets P^1(x_1),P^1(x_2)
\end{align}
For $2\leq c \leq d$, we write rules analogous to the ones above, but instead of a single parameter $x$, we have parameters $x_1,\cdots, x_c$. 
For example, we generalize rule \ref{subs-rule-par-1} by concatenating an open parenthesis before $x_1$ and the corresponding close parenthesis after $x_c$. 
Formally, we have the following rules.
\begin{align}
&P^{c}(\epsilon,\cdots, \epsilon)  \label{subs-rule-k-epsilon} \\
&P^{c}(\pcOne{\op_i} x_1, x_2, \cdots, x_{{c}-1}, x_{c}  \pcOne{\cp_i}) \gets P^{c}(x_1, \cdots, x_{c}) \qquad \forall i \in [k]\label{subs-rule-k-parentheses}
\end{align}
For concatenation, the concatenated strings may have already been ``split'' into parameters of some predicate, so we must add concatenation rules for all possible combinations of already ``split'' strings. Formally, for all $a,b,c$ such that $1\leq a,b,c \leq d$ and $a+b=c+1$, we have the rule:
\begin{align}
P^{c}(x_1, \cdots , x_a y_1, \cdots, y_b) &\gets P^a(x_1, \cdots, x_a), P^b(y_1, \cdots, y_b)\label{subs-rule-k-concatenation}
\end{align}

The following lemma states the completeness property of the predicates $P^c$.
\begin{restatable}{lemma}{lemsubslemmaseparation}\label{lem:subs-lemma-separation}
For every $c\leq d$ and $x_1x_2\cdots x_{c} \in \Dyck_\alpha$, we have that $P^{c}\Derives(x_1,x_2,\cdots,x_{c})$.
\end{restatable}
\SubParagraph{Example.} Consider the string $\pcOne{\op_1\cp_1\op_2\op_1\cp_1\cp_2}$ in $\Dyck_\alpha$. We will arbitrarily partition it into $\{\pcOne{\op_1},\pcOne{\cp_1\op_2\op_1\cp_1\cp_2}\}$. Lemma \ref{lem:subs-lemma-separation} states that $P^{2}\Derives(\pcOne{\op_1},\pcOne{\cp_1\op_2\op_1\cp_1\cp_2}\})$. This is indeed the case, as $P^{2}(\pcOne{\op_1},\pcOne{\cp_1\op_2\op_1\cp_1\cp_2}\})$ can be derived from $P^{2}(\pcOne{\op_1},\pcOne{\cp_1\op_2\op_1\cp_1\cp_2}\})\gets^{\ref{subs-rule-k-concatenation}} P^{2}(\pcOne{\op_1},\pcOne{\cp_1}), P^{1}(\pcOne{\op_2\op_1\cp_1\cp_2})$. In turn, $P^{2}(\pcOne{\op_1},\pcOne{\cp_2})\gets^{\ref{subs-rule-k-parentheses}}P^{2}(\epsilon,\epsilon)\gets^{\ref{subs-rule-k-epsilon}}\epsilon$ and $P^{1}(\pcOne{\op_2\op_1\cp_1\cp_2})\gets^{\ref{subs-rule-par-1}}P^{1}(\pcOne{\op_1\cp_1})\gets^{\ref{subs-rule-par-1}}P^{1}(\epsilon)\gets^{\ref{subs-rule-eps-1}}\epsilon$.

The rules for the predicates $Q^c$ are defined analogously, capturing strings in $\Dyck_{\beta}$.
\begin{align}
&Q^1(\epsilon) \label{rule:eps_t_1} \\
&Q^1(\pcTwo{\ob_i} x \pcTwo{\cb_i}) \gets Q^1(x) \qquad \forall i \in [k] \label{rule:par_t_1}\\        
&Q^1(x_1x_2) \gets Q^1(x_1),Q^1(x_2)
\end{align}

For $2\leq c \leq d$, we have the following rules.
\begin{align}
& Q^c(\epsilon, \cdots, \epsilon) \label{rule:eps_t}\\
&Q^c(\pcTwo{\ob_i} x_1, x_2, \cdots, x_{c-1}, x_c  \pcTwo{\cb_i}) \gets Q^c(x_1,x_2, \cdots, x_c) \qquad \forall i \in [k] \label{rule:par_t}
\end{align}

Finally, for all $a,b,c$ with $1\leq a,b,c \leq d$ and $a+b=c+1$, we have the following concatenation rule.
\begin{align}
Q^c(x_1, \cdots , x_a y_1, \cdots, y_b) &\gets Q^a(x_1,x_2 \cdots, x_a), Q^b(y_1,y_2 \cdots, y_b)\label{rule:conc_t}
\end{align}

The following lemma states the completeness property of the predicates $Q^c$.

\begin{restatable}{lemma}{lemsubslemmaseparationt}\label{lem:subs-lemma-separation-t}
For every $c\leq d$ and $x_1x_2\cdots x_{c} \in \Dyck_\beta$, we have that $Q^{c}\Derives(x_1,x_2,\cdots,x_{c})$.
\end{restatable}
\SubParagraph{Example.} Analogously to the example for \cref{lem:subs-lemma-separation}, we may show that $Q^{2}\Derives(\pcTwo{\ob_1},\pcTwo{\cb_1\ob_2\ob_1\cb_1\cb_2}\})$ by showing that $Q^{2}(\pcTwo{\ob_1},\pcTwo{\cb_1\ob_2\ob_1\cb_1\cb_2}\})\gets^{\ref{rule:conc_t}} Q^{2}(\pcTwo{\ob_1},\pcTwo{\cb_1}), Q^{1}(\pcTwo{\ob_2\ob_1\cb_1\cb_2})$ and, in turn, 
 $Q^{2}(\pcTwo{\ob_1},\pcTwo{\cb_1})\gets^{\ref{rule:par_t}}Q^{2}(\epsilon,\epsilon)\gets^{\ref{rule:eps_t}}\epsilon$ and $Q^{1}(\pcTwo{\ob_2\ob_1\cb_1\cb_2})\gets^{\ref{rule:par_t_1}}Q^{1}(\pcTwo{\ob_1\cb_1})\gets^{\ref{rule:par_t_1}}Q^{1}(\epsilon)\gets^{\ref{rule:eps_t_1}}\epsilon$.


\Paragraph{Interleaving $Q$ and $P$.} 
The last step for our grammar $\Grammar_d^{\circ}$ is to define the interleaving rules deriving the start symbol $\StartSymbol$. Our final string will be formed by mixing segments of a string $s_{\alpha}\in \Dyck_{\alpha}$ and a string $s_{\beta}\in \Dyck_{\beta}$.
\begin{align}
\StartSymbol(x_1 y_1 \cdots x_d y_d) &\gets P^d(x_1, \cdots, x_d), Q^d(y_1, \cdots, y_d)\label{subs-rule-k-interleave}\\
\StartSymbol(y_1 x_1 \cdots y_d x_d) &\gets P^d(x_1, \cdots, x_d), Q^d(y_1, \cdots, y_d)\label{subs-rule-k-interleave-2}
\end{align}

\begin{figure}
\centering
\begin{tikzpicture}[node distance={15mm}, line width=1pt, main/.style = {draw, rectangle},scale=0.8] 
    \node[main] at (-1.5,1.5) (0) [white]{\footnotesize nada}; 
    \node[main] at (0.3,1.5) (1) [black, text=black]{\footnotesize$P^2(\epsilon,\epsilon)$}; 
    \node[main] at (-3.7,0.5) (2) [white]{\footnotesize nada}; 
    \node[main] at (-1.9,0.5) (3) [black, text=black]{\footnotesize$P^2(\epsilon,\epsilon)$}; 
    \node[main] at (0.2,0.5) (4) [black, text=black]{\footnotesize$P^2(\pcOne{(_1},\pcOne{)_1})$}; 
    \node[main] at (2.95,0.5) (5) [black, text=black]{\footnotesize$P^3(\pcOne{(_1},\pcOne{)_1},\epsilon)$}; 
    \node[main] at (-0.8,-0.5) (6) [white]{\footnotesize nada}; 
    \node[main] at (1,-0.5) (7) [black, text=black]{\footnotesize$P^2(\epsilon,\epsilon)$}; 
    \node[main] at (3.1,-0.5) (8) [black, text=black]{\footnotesize$P^2(\pcOne{(_2},\pcOne{)_2})$}; 
    \node[main] at (6.3,0) (9) [black, text=black]{\footnotesize$P^4(\pcOne{(_1},\pcOne{)_1},\pcOne{(_2},\pcOne{)_2})$}; 
    \node[main] at (-8.4,-1.5) (10) [white]{};
    \node[main] at (-7.1,-1.5) (11) [black, text=black]{\footnotesize$Q^1(\epsilon)$}; 
    \node[main] at (-6.9,-0.5) (12) [white]{\footnotesize$nada$};
    \node[main] at (-5.1,-0.5) (13) [black, text=black]{\footnotesize$Q^2(\epsilon,\epsilon)$}; 
    \node[main] at (-4.3,-2) (16) [white]{\footnotesize$nada$};
    \node[main] at (-5.2,-1.5) (14) [black, text=black]{\footnotesize$Q^1(\pcTwo{[_2}\pcTwo{]_2})$}; 
    \node[main] at (-2.7,-1) (15) [black, text=black]{\footnotesize$Q^2(\epsilon,\pcTwo{[_2}\pcTwo{]_2})$}; 
    \node[main] at (-2.5,-2) (17) [black, text=black]{\footnotesize$Q^2(\epsilon,\epsilon)$}; 
    \node[main] at (-0,-1.5) (18) [black, text=black]{\footnotesize$Q^3(\epsilon,\pcTwo{[_2}\pcTwo{]_2},\epsilon)$}; 
    \node[main] at (2.8,-1.5) (19) [black, text=black]{\footnotesize$Q^3(\pcTwo{[_1},\pcTwo{[_2}\pcTwo{]_2},\pcTwo{]_1})$}; 
    \node[main] at (1.5,-2.5) (20) [white]{\footnotesize$nada$};
    \node[main] at (3.3,-2.5) (21) [black, text=black]{\footnotesize$Q^2(\epsilon,\epsilon)$}; 
    \node[main] at (6.2,-2) (22) [black, text=black]{\footnotesize$Q^4(\pcTwo{[_1},\pcTwo{[_2}\pcTwo{]_2},\pcTwo{]_1},\epsilon)$};  
    \node[main] at (7.5,-1) (23) [black, text=black]{\footnotesize$\StartSymbol(\pcOne{(_1}\pcTwo{[_1}\pcOne{)_1}\pcTwo{[_2}\pcTwo{]_2}\pcOne{(_2}\pcTwo{]_1}\pcOne{)_2})$};  
    \draw [->] (0) -- (1)  node [midway,above] {\footnotesize \ref{subs-rule-k-epsilon}};
    \draw [->] (2) -- (3)  node [midway,above] {\footnotesize \ref{subs-rule-k-epsilon}};
    \draw [->] (6) -- (7)  node [midway,above] {\footnotesize \ref{subs-rule-k-epsilon}};
    \draw [->] (10) -- (11)  node [midway,above] {\footnotesize \ref{rule:eps_t}};
    \draw [->] (12) -- (13)  node [midway,above] {\footnotesize \ref{rule:eps_t}};
    \draw [->] (16) -- (17)  node [midway,above] {\footnotesize \ref{rule:eps_t}};
    \draw [->] (20) -- (21)  node [midway,above] {\footnotesize \ref{rule:eps_t}};
    \draw [->] (7) -- (8)  node [midway,above] {\footnotesize \ref{subs-rule-k-parentheses}};
    \draw [->] (3) -- (4)  node [midway,above] {\footnotesize \ref{subs-rule-k-parentheses}};
    \draw [->] (18) -- (19)  node [midway,above] {\footnotesize \ref{rule:par_t}};

    \draw [->] (11) -- (14)  node [midway,above] {\footnotesize \ref{rule:par_t_1}};

    \draw [->] (13) -- (-4.2,-0.5) -- (-4.2,-1) --  (15)  node [midway,above] {\footnotesize\ref{rule:conc_t}};
    \draw [] (14) -- (-4.2,-1.5) -- (-4.2,-1);

     \draw [->] (15) -- (-1.6,-1) -- (-1.6,-1.5) --  (18)  node [midway,above] {\footnotesize\ref{rule:conc_t}};
    \draw [] (17) -- (-1.6,-2) -- (-1.6,-1.5);
    
    \draw [->] (1) -- (1.3,1.5) -- (1.3,0.5) --  (5)  node [midway,above] {\footnotesize\ref{subs-rule-k-concatenation}};
    \draw [] (4) -- (1.3,0.5);
    
    \draw [->] (5) -- (4.2,0.5) -- (4.2,0) --  (9)  node [midway,above] {\footnotesize\ref{subs-rule-k-concatenation}};
    \draw [] (8) -- (4.2,-0.5) -- (4.2,0);
    
    \draw [->] (19) -- (4.2,-1.5) -- (4.2,-2) --  (22)  node [midway,above] {\footnotesize\ref{rule:conc_t}};
    \draw [] (21) -- (4.2,-2.5) -- (4.2,-2);
    
    \draw [->] (9) -- (5.3,-0) -- (5.3,-1) --  (23)  node [midway,above] {\footnotesize \ref{subs-rule-k-interleave}};
    \draw [] (22) -- (5.3,-2) -- (5.3,-1);
    \node[main] at (6.3,0) (a) [black, text=black,fill=white]{\footnotesize$P^4(\pcOne{(_1},\pcOne{)_1},\pcOne{(_2},\pcOne{)_2})$}; 
    \node[main] at (6.2,-2) (aa) [black, text=black, fill=white]{\footnotesize$Q^4(\pcTwo{[_1},\pcTwo{[_2}\pcTwo{]_2},\pcTwo{]_1},\epsilon)$};  
\end{tikzpicture}
\caption{
A derivation tree witnessing $\pcOne{(_1}\pcTwo{[_1}\pcOne{)_1}\pcTwo{[_2}\pcTwo{]_2}\pcOne{(_2}\pcTwo{]_1}\pcOne{)_2}\in \Language(\Grammar_4)$.
}
\label{fig:mcfl_dyck_simple}
\end{figure}

\begin{figure}
\begin{subfigure}[t]{0.3\textwidth}
\centering
\begin{tikzpicture}[node distance={15mm}, line width=1pt, main/.style = {draw, rectangle},scale=0.8] 
    \node[main] at (0,0) (p41) [white, text=black]{\footnotesize$P^4(\pcOne{x_1}, \pcOne{x_2},\pcOne{x_3},\pcOne{x_4})$}; 
    \node[main] at (3,0) (q41) [white, text=black]{\footnotesize$Q^4(\pcTwo{y_1}, \pcTwo{y_2},\pcTwo{y_3},\pcTwo{y_4})$};  
    \node[main] at (1.3,-1.5) (s1) [white, text=black]{\footnotesize$\StartSymbol(\pcOne{x_1}\ \ \pcTwo{y_1}\ \ \pcOne{x_2}\ \ \pcTwo{y_2}\ \ \pcOne{x_3}\ \ \pcTwo{y_3}\ \ \pcOne{x_4}\ \ \pcTwo{y_4})$};  
    \path (-0.5,-0.3) edge [out=-90,in=90, mybluecolor,->](-0.3, -1.35);
    \path (-0,-0.3) edge [out=-80,in=90, mybluecolor,->](0.6, -1.35);
    \path (0.5,-0.3) edge [out=-70,in=90, mybluecolor,->](1.65, -1.35);
    \path (1,-0.3) edge [out=-60,in=90, mybluecolor,->](2.65, -1.35);
    \path (2.4,-0.3) edge [out=-150,in=90, \darkred,->](0.15, -1.35);
    \path (2.9,-0.3) edge [out=-140,in=90, \darkred,->](1.2, -1.35);
    \path (3.4,-0.3) edge [out=-130,in=90, \darkred,->](2.2, -1.35);
    \path (3.9,-0.3) edge [out=-120,in=90, \darkred,->](3.2, -1.35);
\end{tikzpicture}
\subcaption{\label{subfig:interleaving_ex}
Interleaving of the substrings parsed by $P^4$ and $Q^4$ to form $\StartSymbol$ through rule \ref{subs-rule-k-interleave}. 
}
\end{subfigure}
\qquad
\begin{subfigure}[t]{0.6\textwidth}
\centering
\begin{tikzpicture}[node distance={15mm}, line width=1pt, main/.style = {draw, rectangle},scale=0.8] 
    \node[main] at (7,0) (p42) [white, text=black]{\footnotesize$P^4(\pcOne{(_1}, \pcOne{)_1},\pcOne{(_2},\pcOne{)_2})$}; 
    \node[main] at (10,0) (q42) [white, text=black]{\footnotesize$Q^4(\pcTwo{[_1}, \pcTwo{[_2]_2},\pcTwo{]_1},\eps)$};  
    \node[main] at (8.3,-1.5) (s2) [white, text=black]{\footnotesize$\StartSymbol(\pcOne{(_1} \ \ \pcTwo{[_1} \ \ \pcOne{)_1}\ \ \pcTwo{[_2}\pcTwo{]_2}\ \ \pcOne{(_2}\ \ \pcTwo{]_1}\ \ \pcOne{)_2}\ \ \eps )$};  
    \path (6.6,-0.3) edge [out=-90,in=90, mybluecolor,->](6.8, -1.2);
    \path (7,-0.3) edge [out=-80,in=90, mybluecolor,->](7.7, -1.2);
    \path (7.5,-0.3) edge [out=-70,in=90, mybluecolor,->](8.9, -1.2);
    \path (7.9,-0.3) edge [out=-60,in=90, mybluecolor,->](9.8, -1.2);
    \path (9.4,-0.3) edge [out=-150,in=90, \darkred,->](7.3, -1.2);
    \path (10,-0.3) edge [out=-130,in=90, \darkred,->](8.3, -1.2);
    \path (10.6,-0.3) edge [out=-120,in=90, \darkred,->](9.4, -1.2);
    \path (11,-0.3) edge [out=-110,in=90, \darkred,->](10.3, -1.2);
\end{tikzpicture}
\subcaption{\label{subfig:interleaving_cc}
Illustration of the interleaving in a concrete example: the final derivation, as shown in \cref{fig:mcfl_dyck_simple}, forms the string $\pcOne{(_1}\pcTwo{[_1}\pcOne{)_1}\pcTwo{[_2}\pcTwo{]_2}\pcOne{(_2}\pcTwo{]_1}\pcOne{)_2}$. 
We show the $\eps$ in the final string for clarity.
}
\end{subfigure}
\caption{
Interleavings of elements in $P^4$ and $Q^4$ to form a string in $\StartSymbol$.
}
\label{fig:example_interleaving}
\end{figure}

For example, the string $\pcOne{(_1}\pcTwo{[_1}\pcOne{)_1}\pcTwo{[_2}\pcTwo{]_2}\pcOne{(_2}\pcTwo{]_1}\pcOne{)_2}\in \IDL$ is in $\Language(\Grammar_4^{\circ})$ (and also in $\Language(\Grammar_d^\circ)$  for all  $d\geq 4$). 
\cref{fig:mcfl_dyck_simple} illustrates a corresponding derivation tree and \cref{fig:example_interleaving} shows in more detail how the elements of each predicate are combined.
The following lemma states that the underapproximation of $\IDL$ via the languages $\Language(\Grammar_d^{\circ})_{d>0}$ is complete in the limit.

\begin{restatable}{lemma}{lemmathereisdimension}\label{lem:there_is_dimension}
For all $s\in \IDL$, there exists a dimension $d_0$ such that for all $d\geq d_0$, we have $s\in \Language(\Grammar_d^{\circ})$.
\end{restatable}

To obtain some indication of the coverage of each grammar $\Grammar_d^{\circ}$, consider the interleaved Dyck language $\IDL$ over the two alphabets $\{ \pcOne{\op_1},\pcOne{\cp_1}\}$ and   $\{\pcTwo{\ob_1},\pcTwo{\cb_1}\}$. 
\cref{tab:simple_mcfl} shows how many strings of $\IDL$ of a fixed size belong to $\Language(\Grammar_d^{\circ})$,
for small values of $d$.

\begin{table}
\caption{\label{tab:simple_mcfl}
The number of strings of a given length of $\IDL$ over the alphabets $\{ \pcOne{\op_1},\pcOne{\cp_1}\}$ and   $\{\pcTwo{\ob_1},\pcTwo{\cb_1}\}$, that are in $\Language(\Grammar_d^\circ)$.
The languages gain expressiveness as the dimension $d$ increases, converging towards $\IDL$.
} 
\centering
\begin{tabular}{|l|r|r|r|r|r|r|r|}
\cline{2-8}
\multicolumn{1}{c}{}& \multicolumn{7}{|c|}{\textbf{Length of string}}\\
\cline{2-8}
\multicolumn{1}{c|}{} & 2 & 4 & 6 & 8 & 10 & 12 & 14\\
\hline
Strings in $\Language(\Grammar_1^\circ)$ & 2 & 6 & 18 & 56 & 180 & 594 & 2,002\\
Strings in $\Language(\Grammar_2^\circ)$ & 2 & 10 & 58 & 312 & 1,556 & 7,358 & 33,546\\
Strings in $\Language(\Grammar_3^\circ)$ & 2 & 10 & 70 & 556 & 4,244 & 29,642 & 190,334\\
\hline
\multicolumn{1}{|l|}{Strings in $\IDL$} & 2 & 10 & 70 & 588 & 5,544 & 56,628 & 613,470\\
\hline
\end{tabular}
\end{table}

\subsection{The Full Grammar $\Grammar_d^{+}$} \label{SUBSEC:EXPRESSIVENESS}

Having established a basic intuition on the grammar $\Grammar_d^{\circ}$,
we now obtain our more expressive grammar $\Grammar_d^{+}$, by introducing some additional production rules
for the predicates $P^c$ and $Q^c$.
Intuitively, $P^c$ (resp., $Q^c$) will no longer be restricted to parsing strings consisting entirely of parentheses (resp., brackets).
Instead, for some $P^c(x_1,\cdots, x_c)$, the substring $x_i$ may contain brackets, as long as $x_i\Project \{\pcTwo{\ob_i},\pcTwo{\cb_i}\}_{i \in [k]} \in \Dyck_\beta$, i.e,  the brackets in each $x_i$, by themselves, are validly nested. 
Analogously, for some $Q^c(x_1,\cdots, x_c)$, the substring $x_i$ may contain parentheses, as long as $x_i\Project \{\pcOne{\op_i},\pcOne{\cp_i}\}_{i \in [k]} \in \Dyck_\alpha$.
This interleaving is based on the following lemma.

\begin{restatable}{lemma}{leminsertstringdyck}\label{lem:insert_string_dyck}
Consider any Dyck language $\Dyck$.
If $s_1s_2\in \Dyck$ and $t\in \Dyck$ then $s_1ts_2\in \Dyck$.
\end{restatable}

In particular, the new production rules for the predicates $P^c$ and $Q^c$ maintain the following invariants.
Consider any sequence of strings $s_1, \dots, s_c$.
\begin{align}
&\text{If }P^c\Derives(s_1,\dots, s_c) \text{ then } (s_1\cdots s_c \Project \{\pcOne{\op_i},\pcOne{\cp_i}\}_{i \in [k]} )\in \Dyck_\alpha \text{ and } \forall j \in [c]. (s_j \Project \{\pcTwo{\ob_i},\pcTwo{\cb_i}\}_{i \in [k]}) \in \Dyck_\beta
\label{eq:rules_invariant_parentheses}\\
&\text{If }Q^c\Derives(s_1,\dots, s_c) \text{ then } (s_1\cdots s_c \Project \{\pcTwo{\ob_i},\pcTwo{\cb_i}\}_{i \in [k]} )\in \Dyck_\beta \text{ and } \forall j \in [c]. (s_j \Project \{\pcOne{\op_i},\pcOne{\cp_i}\}_{i \in [k]}) \in \Dyck_\alpha \label{eq:rules_invariant_brackets}
\end{align}
In words, the invariant in \cref{eq:rules_invariant_parentheses} (resp., \cref{eq:rules_invariant_brackets}) guarantees that $P^c$ (resp.,  $Q^c$) parses, as a whole, a valid string of parenthesis (resp.,  brackets) interleaved with multiple valid strings of brackets (resp.,  parenthesis), each of these being fully contained in some $s_i$.
Together, these invariants guarantee a sound interleaving of the corresponding substrings, as stated in the following lemma.

\begin{restatable}{lemma}{leminterdyckmcfginvariants}\label{lem:inter_dyck_mcfg_invariants}
Assume that for each $c\in [d]$, the predicates $P^c$ (resp., $Q^c)$ satisfy the invariant in \cref{eq:rules_invariant_parentheses} (resp., \cref{eq:rules_invariant_brackets}).
Then any string $s$ such that $\StartSymbol\Derives(s)$ satisfies $s\in \IDL$.
\end{restatable}

It is straightforward to verify that the production rules defined in \cref{SUBSEC:BASIC_RULES} follow the invariants in \cref{eq:rules_invariant_parentheses} and \cref{eq:rules_invariant_brackets}.
We now further add two classes of production rules, called \emph{insertion} and \emph{nesting}, that increase the expressiveness of the grammar while maintaining the invariants.

\Paragraph{Insertion.} 
Since we know that every string identified by $\StartSymbol$ is in $\IDL$, we can, from \cref{lem:insert_string_dyck}, insert any string $y$ such that $\StartSymbol(y)$ in between any characters of a valid string $x_1\cdots x_c$. The insertion rules concatenate strings derived by $\StartSymbol$ to parameters of $P^c$ and $Q^c$. 
For each $1\leq c \leq d$ and $1 \leq i \leq c$, we have the following rules.
\begin{align}
P^c(x_1, \dots , yx_i ,\dots, x_{c}) &\gets  P^{c}(x_1, \dots, x_{c}), \StartSymbol(y) \label{subs-rule-k-conc-1}\\
P^c(x_1, \dots , x_iy ,\dots, x_{c}) &\gets  P^{c}(x_1, \dots, x_{c}), \StartSymbol(y) \label{subs-rule-k-conc-2}\\
Q^c(x_1, \dots , yx_i ,\dots, x_{c}) &\gets  Q^{c}(x_1, \dots, x_{c}), \StartSymbol(y) \\
Q^c(x_1, \dots , x_iy ,\dots, x_{c}) &\gets  Q^{c}(x_1, \dots, x_{c}), \StartSymbol(y) \label{rule:tts}
\end{align}

Notice that the invariants \cref{eq:rules_invariant_parentheses} and \cref{eq:rules_invariant_brackets} are maintained because all the insertions follow the conditions in \cref{lem:insert_string_dyck}. 
In other words, since we are only inserting strings $y \in \IDL$, every string $s$ that now encloses $y$ (i.e., derived by the predicates on the left-hand side) is valid on the parentheses or brackets ($ s \Project \{\pcOne{\op_i},\pcOne{\cp_i}\}_{i \in [k]} \in \Dyck_\alpha$ or $ s \Project \{\pcTwo{\ob_i},\pcTwo{\cb_i}\}_{i \in [k]} \in \Dyck_\beta $ ), as long as this was already the case without $y$ (i.e., derived by the $P^c$/$Q^c$ predicates on the right-hand side).

\SubParagraph{Example.} 
To illustrate the increase in expressiveness with the insertion rules, consider the string $\pcOne{(_1}\pcTwo{[_1}\pcOne{)_1}\pcTwo{[_2}\pcTwo{]_2}\pcOne{(_2}\pcTwo{]_1}\pcOne{)_2}$.
It is derived by $\Grammar_3^{+}$, while it is only derived by $\Grammar_d^{\circ}$ for $d\geq 4$.
This can be seen, for example, through the application of Rule~\ref{subs-rule-k-interleave} to the predicates $P^3(\pcOne{(_1},\pcOne{)_1}\pcTwo{[_2}\pcTwo{]_2}\pcOne{(_2},\pcOne{)_2})$ and $Q^3(\pcTwo{[_1},\pcTwo{]_1},\epsilon)$. The predicate $P^3(\pcOne{(_1},\pcOne{)_1}\pcTwo{[_2}\pcTwo{]_2}\pcOne{(_2},\pcOne{)_2})$ can be generated by concatenating $P^2(\pcOne{(_1},\pcOne{)_1}\pcTwo{[_2}\pcTwo{]_2})$ and  $P^2(\pcOne{(_2},\pcOne{)_2})$ through the application of Rule~\ref{subs-rule-k-concatenation}. $P^2(\pcOne{(_1},\pcOne{)_1}\pcTwo{[_2}\pcTwo{]_2})$, in turn, is generated by the application of Rule~\ref{subs-rule-k-conc-2} to $P^2(\pcOne{(_1},\pcOne{)_1})$ and $\StartSymbol(\pcTwo{[_2}\pcTwo{]_2})$, with $i=2$.

\Paragraph{Nesting.} Nesting allows for more complex ``mixing'' of strings in $\Dyck_\alpha$ and $\Dyck_\beta$. Essentially, we want to add some interleaving without relying exclusively on rules \ref{subs-rule-k-interleave} and \ref{subs-rule-k-interleave-2}. One way to do this would be to add a rule that derives $P^c(x_1, \dots , \pcTwo{\ob_j} x_i \pcTwo{\cb_j} ,\dots, x_{k})$ from $P^c(x_1, \dots , x_i,\dots, x_{k})$. We, instead, add a more general version of this rule, by enclosing $x_i$ in any $\{y_1,y_2\}$ such that $Q^2(y_1,y_2)$ as opposed to only $\{\pcTwo{\ob_j},\pcTwo{\cb_j}\}$. We also add the equivalent rules for the predicates $Q^c$ and $P^2$.
The nesting rules thus allow the enclosing of each parameter in predicates of dimension two. 
Notice that this rule only applies when $d>2$. 
For each $1\leq c \leq d$ and $1 \leq i \leq c$, we have the following rules.
\begin{align}
P^c(x_1, \dots , y_1x_iy_2 ,\dots, x_{k}) &\gets  P^{c}(x_1, \dots, x_{k}), Q^2(y_1,y_2) \label{subs-rule-k-enclose-parameter-p}\\
Q^c(x_1, \dots , y_1x_iy_2 ,\dots, x_{k}) &\gets  Q^{c}(x_1, \dots, x_{k}), P^2(y_1,y_2)\label{subs-rule-k-enclose-parameter-q}
\end{align}

We similarly add variants of Rule~\ref{subs-rule-k-parentheses} (resp., Rule~\ref{rule:par_t}) that generalize the enclosing of a predicate in $(\pcOne{\op_i},\pcOne{\cp_i})$ (resp. $(\pcTwo{\ob_i},\pcTwo{\cb_i})$) to a more general $(y_1,y_2)$ for some $P^2(y_1,y_2)$ (resp. $Q^2(y_1,y_2)$).  
For $1\leq c \leq d$
\begin{align}
&P^{c}(y_1 x_1, x_2, \dots, x_{{c}-1}, x_{c} y_2) \gets P^{c}(x_1, \dots, x_{c}), P^2(y_1,y_2)\\
&Q^{c}(y_1 x_1, x_2, \dots, x_{{c}-1}, x_{c} y_2) \gets Q^{c}(x_1, \dots, x_{c}), Q^2(y_1,y_2)
\end{align}

The invariants \cref{eq:rules_invariant_parentheses} and \cref{eq:rules_invariant_brackets} are also preserved under the nesting rules.

\SubParagraph{Examples.} 
Consider again the string $\pcOne{(_1}\pcTwo{[_1}\pcOne{)_1}\pcTwo{[_2}\pcTwo{]_2}\pcOne{(_2}\pcTwo{]_1}\pcOne{)_2}$.
Using the nesting rules, it can also be produced by $\Grammar_2$. 
This can be done, for example, through the application of Rule~\ref{subs-rule-k-interleave-2} in the predicates $P^2(\pcTwo{[_2}\pcTwo{]_2}, \epsilon)$ and $Q^2(\pcOne{(_1}\pcTwo{[_1}\pcOne{)_1},\pcOne{(_2}\pcTwo{]_1}\pcOne{)_2})$.
$Q^2(\pcOne{(_1}\pcTwo{[_1}\pcOne{)_1},\pcOne{(_2}\pcTwo{]_1}\pcOne{)_2})$ can be generated through the application of Rule~\ref{subs-rule-k-enclose-parameter-q} to $Q^2(\pcTwo{[_1},\pcOne{(_2}\pcTwo{]_1}\pcOne{)_2})$ and $P^2(\pcOne{(_1},\pcOne{)_1})$ with $i=1$. In turn, $Q^2(\pcTwo{[_1},\pcOne{(_2}\pcTwo{]_1}\pcOne{)_2})$  can be generated through the application of Rule~\ref{subs-rule-k-enclose-parameter-q} to $Q^2(\pcTwo{[_1},\pcTwo{]_1})$ and $P^2(\pcOne{(_1},\pcOne{)_1})$ with $i=2$.

\begin{table}
\caption{\label{tab:final_mcfl}
The number of strings of a given length of $\IDL$ over the alphabets $\{ \pcOne{\op_1},\pcOne{\cp_1}\}$ and   $\{\pcTwo{\ob_1},\pcTwo{\cb_1}\}$, that are in $\Language(\Grammar_d^{+})$.
The languages gain expressiveness as the dimension $d$ increases, converging towards $\IDL$.
Compared to \cref{tab:simple_mcfl}, the new derivation rules significantly increase the expressiveness of the grammar.
} 
\centering
\begin{tabular}{|l|r|r|r|r|r|r|r|}
\cline{2-8}
\multicolumn{1}{c}{}& \multicolumn{7}{|c|}{Length of string}\\
\cline{2-8}
\multicolumn{1}{c|}{} & 2 & 4 & 6 & 8 & 10 & 12 & 14\\
\hline
Strings in $\Language(\Grammar_1^+)$ & 2 & 8 & 40 & 224 & 1,344 & 8,448 & 54,912\\
Strings in $\Language(\Grammar_2^+)$ & 2 & 10 & 70 & 588 & 5,544 & 56,612 & 612,686\\
Strings in $\Language(\Grammar_3^+)$ & 2 & 10 & 70 & 588 & 5,544 & 56,628 & 613,470\\
\hline
\multicolumn{1}{|l|}{Strings in $\IDL$} & 2 & 10 & 70 & 588 & 5,544 & 56,628 & 613,470\\
\hline
\end{tabular}
\end{table}

\cref{tab:final_mcfl} shows how many strings of $\IDL$, over the alphabets $\{ \pcOne{\op_1},\pcOne{\cp_1}\}$ and $\{\pcTwo{\ob_1},\pcTwo{\cb_1}\}$, and of a given size belong to $\Language(\Grammar_d^+)$, for different values of $d$.
Compared to \cref{tab:simple_mcfl}, we see that the expressiveness has increased significantly.

\begin{figure}[!h] 
\begin{subfigure}[t]{0.37\textwidth}
\centering
\begin{tikzpicture}[node distance={15mm}, line width=1pt, main/.style = {draw, rectangle},scale=0.8] 
    \node[main] at (1.5,-0.5) (aa) [white]{}; 
    \node at (-3.5,2.5) (0) [white]{\footnotesize}; 
    \node[main] at (-1.9,2.5) (1) [black, text=black]{\footnotesize$P^2(\epsilon,\epsilon)$}; 
    \node[main] at (0.3,2.5) (15) [black, text=black]{\footnotesize$P^2(\pcOne{(_{10}},\pcOne{)_{10}})$}; 
    \node at (-3.5,0.5) (2) [white]{\footnotesize}; 
    \node[main] at (-1.9,0.5) (3) [black, text=black]{\footnotesize$Q^2(\epsilon,\epsilon)$}; 
    \node[main] at (0.3,0.5) (4) [black, text=black]{\footnotesize$Q^2(\pcTwo{[_1},\pcTwo{]_1})$}; 
    \node[main] at (1.5,1.5) (5) [black, text=black]{\footnotesize$\StartSymbol(\pcOne{(_{10}}\pcTwo{[_1}\pcOne{)_{10}}\pcTwo{]_1})$}; 
 
    \draw [->] (0) -- (1)  node [midway,above] {\footnotesize \ref{subs-rule-k-epsilon}};
    \draw [->] (1) -- (15)  node [midway,above] {\footnotesize \ref{subs-rule-k-parentheses}};
    \draw [->] (2) -- (3)  node [midway,above] {\footnotesize \ref{rule:eps_t}};
    
    \draw [->] (3) -- (4)  node [midway,above] {\footnotesize \ref{rule:par_t}};

    \draw [->] (15) -- (-0.3,2.5) -- (-0.3,1.5) --  (5)  node [midway,above] {\footnotesize \ref{subs-rule-k-interleave}};
    \draw [] (4) -- (-0.3,0.5) -- (-0.3,1.5);
    \node[main] at (0.3,2.5) (a) [black, text=black,fill=white]{\footnotesize$P^2(\pcOne{(_{10}},\pcOne{)_{10}})$}; 
    \node[main] at (0.3,0.5) (aa) [black, text=black, fill=white]{\footnotesize$Q^2(\pcTwo{[_1},\pcTwo{]_1})$}; 
\end{tikzpicture}
\subcaption{\label{subfig:program_graph_derivation}
A derivation witnessing $\pcOne{(_{10}}\pcTwo{[_1}\pcOne{)_{10}}\pcTwo{]_1}\in \Language(\Grammar_2^+)$, 
capturing dataflow from $a$ to $c$ in  the graph of \cref{subfig:program_graph}.
}
\end{subfigure}
\hspace{0.25cm}
\begin{subfigure}[t]{.6\linewidth}
\begin{tikzpicture}[node distance={15mm}, line width=1pt, main/.style = {draw, rectangle},scale=0.8] 
    \node[main] at (-3.3,2.5) (0) [white]{\footnotesize}; 
    \node[main] at (-1.9,2.5) (1) [black, text=black]{\footnotesize$P^2(\epsilon,\epsilon)$}; 
    \node[main] at (0.55,2.5) (15) [black, text=black]{\footnotesize$P^2(\pcOne{(_{100}},\pcOne{)_{100}})$}; 
    \node[main] at (-3.3,1.7) (00) [white]{\footnotesize}; 
    \node[main] at (-1.9,1.7) (11) [black, text=black]{\footnotesize$Q^2(\epsilon,\epsilon)$}; 
    \node[main] at (0.3,1.7) (115) [black, text=black]{\footnotesize$Q^2(\pcTwo{[_1},\pcTwo{]_1})$}; 

    \node[main] at (3.8,2.1) (p2) [black, text=black]{\footnotesize$P^2(\pcTwo{[_1}\pcOne{(_{100}}\pcTwo{]_1},\pcOne{)_{100}})$}; 
    \node[main] at (-3.3,0.9) (2) [white]{\footnotesize}; 
    \node[main] at (-1.9,0.9) (3) [black, text=black]{\footnotesize$Q^2(\epsilon,\epsilon)$}; 
    \node[main] at (0.3,0.9) (4) [black, text=black]{\footnotesize$Q^2(\pcTwo{[_1},\pcTwo{]_1})$}; 
    \node[main] at (3.2,0.5) (q2) [black, text=black]{\footnotesize$Q^2(\pcTwo{[_1}\pcTwo{[_1}\pcTwo{]_1},\pcTwo{]_1})$}; 
    \node[main] at (-3.3,-0.7) (22) [white]{\footnotesize}; 
    \node[main] at (-1.9,-0.7) (33) [black, text=black]{\footnotesize$Q^2(\epsilon,\epsilon)$}; 
    \node[main] at (0.3,-0.7) (44) [black, text=black]{\footnotesize$Q^2(\pcTwo{[_1},\pcTwo{]_1})$}; 
    \node[main] at (-2.9,0.1) (222) [white]{\footnotesize}; 
    \node[main] at (-1.5,0.1) (333) [black, text=black]{\footnotesize$P^2(\epsilon,\epsilon)$}; 
    \node[main] at (0.8,0.1) (s) [black, text=black]{\footnotesize$\StartSymbol(\pcTwo{[_1}\pcTwo{]_1})$}; 
    \node[main] at (5,1.3) (5) [black, text=black]{\footnotesize$\StartSymbol(\pcTwo{[_1}\pcOne{(_{100}}\pcTwo{]_1}\pcTwo{[_1}\pcTwo{[_1}\pcTwo{]_1}\pcOne{)_{100}}\pcTwo{]_1})$}; 
 
    \draw [->] (0) -- (1)  node [midway,above] {\footnotesize \ref{subs-rule-k-epsilon}};
    \draw [->] (00) -- (11)  node [midway,above] {\footnotesize \ref{rule:eps_t}};
    \draw [->] (2) -- (3)  node [midway,above] {\footnotesize \ref{rule:eps_t}};
    \draw [->] (22) -- (33)  node [midway,above] {\footnotesize \ref{rule:eps_t}};
    \draw [->] (222) -- (333)  node [midway,above] {\footnotesize \ref{subs-rule-k-epsilon}};
    \draw [->] (1) -- (15)  node [midway,above] {\footnotesize \ref{subs-rule-k-parentheses}};
    \draw [->] (11) -- (115)  node [midway,above] {\footnotesize \ref{rule:par_t}};
    \draw [->] (3) -- (4)  node [midway,above] {\footnotesize \ref{rule:par_t}};
    \draw [->] (33) -- (44)  node [midway,above] {\footnotesize \ref{rule:par_t}};

    \draw [->] (333) -- (s)  node [near end,above=-1pt] {\footnotesize \ref{subs-rule-k-interleave}};
    \draw [] (-0.3,-0.4) -- (-0.3,0.1);

    \draw [->] (4) -- (1.4,0.9) -- (1.4,0.5) --  (q2)  node [midway,above] {\footnotesize \ref{rule:tts}};
    \draw [] (1.4,0.4) -- (1.4,0.5);

    \draw [->] (15) -- (1.8,2.5) -- (1.8,2.1) --  (p2)  node [midway,above] {\footnotesize \label{subs-rule-k-enclose-parameter-p}};
    \draw [] (115) -- (1.8,1.7) -- (1.8,2.1);

    \draw [->] (2.5,1.8) -- (2.5,1.3) --  (5)  node [midway,above=-0.8pt] {\footnotesize \ref{subs-rule-k-interleave}};
    \draw [] (2.5,0.8) -- (2.5,1.3);

    

\end{tikzpicture}
\subcaption{\label{subfig:real_graph_derivation}
A derivation witnessing $\pcTwo{[_1}\pcOne{(_{100}}\pcTwo{]_1}\pcTwo{[_1}\pcTwo{[_1}\pcTwo{]_1}\pcOne{)_{100}}\pcTwo{]_1}\in \Language(\Grammar_2^+)$, capturing dataflow from $e$ to $k$ in the graph of \cref{subfig:real_graph}.
}
\end{subfigure}
\caption{
Interleaved Dyck reachability witnessed as membership in $\Language(\Grammar_2^{+})$.
}
\label{fig:example_derivation_1}
\end{figure}

Finally, we revisit the two dataflow graphs from \cref{fig:example_context_field_graph}.
\cref{fig:example_derivation_1} illustrates how $\Grammar_2^+$ discovers the corresponding dataflow facts.
\section{MCFL Reachability}\label{SEC:MCFG_REACHABILITY}

Having established the concepts of MCFL reachability (\cref{SEC:PRELIMINARIES}) and its modeling power in static analyses (\cref{SEC:MODELING}), here we address the core algorithmic question, i.e., given a graph $G$ and an $\MCFG$ $\Grammar$, compute $\Language(\Grammar)$ reachability on $G$.
We prove the following theorem.

\thmupperbound*

We proceed towards \cref{thm:upper_bound} as follows.
In \cref{SUBSEC:NORMAL_FORM}, we develop a normal form for $\MCFG$s that allows for an algorithm that is both simpler, and of smaller complexity, compared to handling arbitrary $\MCFG$s.
Then, in \cref{SUBSEC:MCFG_ALGORITHM} we present the main algorithm.
In \cref{SUBSEC:EXAMPLE} we illustrate the algorithm on a small but non-trivial example.
Finally, in \cref{SUBSEC:CORRECTNESS_COMPLEXITY} we establish the correctness and complexity of the algorithm, thereby concluding \cref{thm:upper_bound}.

\subsection{A Normal Form for MCFGs}\label{SUBSEC:NORMAL_FORM}

It is natural to restrict production rules as to make algorithms simpler to design and analyze. 
For CFGs, this is usually done through the Chomsky normal form.
For MCFGs, there exists a normal form~\cite{Seki1991}, which helps solve the MCFL membership problem efficiently.
Unfortunately, that normal form is insufficient for MCFL reachability, in the sense that it would make the natural algorithm for the problem have complexity higher than necessary.
Here we develop a normal form that is suitable for this task.

\Paragraph{Normal form.}
We say that a $\drMCFG{d}{r}$ $\Grammar=(\NonTerminals, \Terminals, \Rules, \StartSymbol)$ is in \emph{normal form} if it is non-deleting and non-permuting (see \cref{SEC:PRELIMINARIES}), and each of its production rules is one of the following types.
\begin{compactdesc}
\item[\emph{Type~1:}] $A(a)$, for some $a\in \{\Terminals\cup \epsilon\}$.
\item[\emph{Type~2:}] $A(x_1,\dots,x_{i-1},ax_i,x_{i+1},\dots,x_k)\gets B(x_1,\dots,x_k)$, for some $a\in \Terminals$.
\item[\emph{Type~3:}] $A(x_1,\dots,x_{i-1},x_ia,x_{i+1},\dots,x_k)\gets B(x_1,\dots,x_k)$, for some $a\in \Terminals$.
\item[\emph{Type~4:}] $A(x_1,\dots, x_{i-1}, a, x_{i},\dots, x_{k})\gets B(x_1,\dots,x_k)$, for some $a\in \{\Terminals \cup \epsilon\}$.
\item[\emph{Type~5:}]
$A_0(s_1,\dots, s_{k_0}) \gets A_1(x_1^1,\dots, x_{k_1}^1) \cdots A_{\ell}(x_1^{\ell},\dots, x_{k_{\ell}}^{\ell})$, each $s_i\in \{ x_j^l \,|\, l\in[\ell],\!j\in [k_l]\}^*$.
\end{compactdesc}

From an operational perspective, the normal form achieves that $\Grammar$ derives a single terminal of the produced string, using rules that have at most one nonterminal on the right-hand side at each rule (Type~1-4), 
while any rule with more than one nonterminals on the right-hand side only merges the underlying substrings (Type~5).
We also note that substrings are allowed to be initialized as $\epsilon$ (Type~1 and Type~4). 
Although this freedom is redundant when parsing strings, it is helpful for solving language reachability, in the sense that various graph edges might be labeled with $\epsilon$, and simply contracting them is not sound (in contrast, when parsing a string, we ignore $\epsilon$ substrings).

The next lemma establishes that normal-form MCFGs do not lose in expressive power.

\begin{restatable}{lemma}{lemnormalform}\label{lem:normal_form}
Given a $\drMCFG{d}{r}$ $\Grammar$ with $r\geq 1$, we can construct in $O(\Poly(|\Grammar|))$ time a 
$\drMCFG{d}{r}$ $\Grammar'$ in normal form such that $\Language(\Grammar)=\Language(\Grammar')$.
\end{restatable}

The transformation steps for obtaining $\Grammar'$ from $\Grammar$ are presented in the proof of \cref{lem:normal_form} in \cref{SEC:APP_MCFG_REACHABILITY}.
\cref{tab:normal_form_example} sketches the transformation process in a small example.

\begin{table}
\setlength\tabcolsep{5mm}
\caption{\label{tab:normal_form_example}
The $\drMCFG{2}{2}$ $\Grammar$ for $\Language=\{w_1w_2\#w_2w_1\mid w_1,w_2\in\{0,1\}^*\}$ from \cref{SEC:PRELIMINARIES} (left) and the corresponding normal-form $\Grammar'$ (right).
Each original rule is replaced by a set of new rules.
}
\centering
\begin{tabular}{cc}
\toprule
\textbf{Original Rule} & \textbf{New Rules}\\
\hline
    \multirow{2}{*}{$A(\epsilon,\epsilon)$} & $A_1(\epsilon)$\\
    & $A(x, \epsilon) \gets A_1(x)$\\
    \hline
    \multirow{2}{*}{$A(x_10,x_20)\gets A(x_1,x_2)$} & $A_2(x_1,x_20) \gets A(x_1, x_2)$\\
    &$A(x_10,x_2) \gets A_2(x_1, x_2)$\\
    \hline
    \multirow{2}{*}{$A(x_11,x_21)\gets A(x_1,x_2)$} & $A_3(x_1,x_21) \gets A(x_1, x_2)$\\
    &$A(x_11,x_2) \gets A_3(x_1, x_2)$\\
    \hline
    \multirow{2}{*}{$S(x_1y_1\#y_2x_2)\gets A(x_1,x_2), A(y_1,y_2)$} & $A_4(x_1\#,x_2) \gets A(x_1, x_2)$\\
    &$S(x_1y_1y_2x_2) \gets  A(x_1,x_2), A_4(y_1, y_2)$\\
\bottomrule
\end{tabular}
\end{table}
\subsection{The MCFL Reachability Algorithm}\label{SUBSEC:MCFG_ALGORITHM}

Given the normal form of \cref{SUBSEC:NORMAL_FORM}, we are now ready to solve MCFL reachability.

\begin{algorithm}[]
\small
\caption{
Algorithm for $\drMCFL{d}{r}$-reachability.
}
\label{algo:mcfl_reach}
\textbf{Input: } A $\drMCFG{d}{r}$ $\Grammar$ in normal form, a labeled graph $G=(V,E,\Label)$\\
\textbf{Output: } All node pairs $(u,v)$ such that $v$ is $\Language(\Grammar)$-reachable from $u$\\
\tcp{Initialization}
\lForEach{$v\in V$}{\label{line:epsilon_edges}
Insert $(v,v)$ with $\Label(v,v)=\epsilon$ in $E$} 
\ForEach{$(u,v)\in E$}{\label{line:algo_type1}
\lForEach{Type~1 rule $A(\Label(u,v))$}{
Insert $A[(u,v)]$ in $\Worklist$ and in $\Done$
}
}
\tcp{Fixpoint computation}
\While{$\Worklist$ is not empty}{\label{line:algo_main_loop}
Extract an element $B[(u_1, v_1),\dots, (u_k,v_k)]$ from $\Worklist$ \label{line:algo_extraction} \\
{
\ForEach{$i \in [k]$}{ \label{line:for_expand_epsilon}
    \ForEach{$(v_i,v)\in E$ with $\Label(v_i,v)=\epsilon$ and $B[(u_1,v_1),\dots, (u_i,v), (u_{i+1},v_{i+1}),\dots (u_k, v_k))]\not \in \Done$}{
        Insert $B[(u_1,v_1),\dots, (u_i,v), (u_{i+1},v_{i+1}),\dots (u_k, v_k))]$ in $\Worklist$ and $\Done$
    }
    \ForEach{$(u,u_i)\in E$ with $\Label(u,u_i)=\epsilon$ and $B[(u_1,v_1),\dots, (u,v_i), (u_{i+1},v_{i+1}),\dots (u_k, v_k))]\not \in \Done$}{
        Insert $B[(u_1,v_1),\dots, (u,v_i), (u_{i+1},v_{i+1}),\dots (u_k, v_k))]$ in $\Worklist$ and $\Done$
    }
}
}
\ForEach{Type~2 rule $A(x_1,\dots,x_{i-1},sx_i,x_{i+1},\dots,x_k)\gets B(x_1,\dots,x_k)$}{\label{line:algo_type2}
\ForEach{$(u,u_i)\in E$ with $\Label(u,u_i)=s$ and $ A[(u_1,v_1),\dots, (u,v_i), (u_{i+1},v_{i+1}),\dots (u_k, v_k))]\not \in \Done$}{
Insert $A[(u_1,v_1),\dots, (u,v_i), (u_{i+1},v_{i+1}),\dots (u_k, v_k))]$ in $\Worklist$ and $\Done$
}
}
\ForEach{Type~3 rule $A(x_1,\dots,x_{i-1},x_is,x_{i+1},\dots,x_k)\gets B(x_1,\dots,x_k)$}{\label{line:algo_type3}
\ForEach{$(v_i,v)\in E$ with $\Label(v_i,v)=s$ and $A[(u_1,v_1),\dots, (v_i,v), (u_{i+1},v_{i+1}),\dots (u_k, v_k))]\not \in \Done$}{
Insert $A[(u_1,v_1),\dots, (v_i,v), (u_{i+1},v_{i+1}),\dots (u_k, v_k))]$ in $\Worklist$ and $\Done$
}
}
\ForEach{Type~4 rule $A(x_1,\dots, x_{i-1}, s, x_{i},\dots, x_{k})\gets B(x_1,\dots,x_k)$}{\label{line:algo_type4}
\ForEach{$(u,v)\in E$ with $\Label(u,v)=s$ and $A[(u_1,v_1),\dots, (u,v), (u_{i},v_{i}),\dots (u_k, v_k))]\not \in \Done$}{
Insert $A[(u_1,v_1),\dots, (u,v), (u_{i},v_{i}),\dots (u_k, v_k))]$ in $\Worklist$ and $\Done$
}
}
\ForEach{Type~5 rule $A_0(s_1,\dots, s_{k_0}) \gets A_1(x_1^1,\dots, x_{k_1}^1) \cdots A_{\ell}(x_1^{\ell},\dots, x_{k_{\ell}}^{\ell})$ with $B=A_l$, for $l\in[\ell]$}{\label{line:algo_type5}
\ForEach{$\{A_{j}[(u^j_1, v^j_1),\dots, (u^j_{k_j}, v^j_{k_j})]\in \Done$, $j\in [\ell]\setminus\{ l \}\}$ such that $s_1\cdots s_{k_0}\PathDerives \{ (u^i_{j_i}, v^i_{j_i}) \}_{i\in[\ell],j_i\in[k_{i}]}$, where $(u^l_{i_l}, v^l_{i_l})=(u_i,v_i)$}{\label{line:algo_type5_pathmatching}
Let $(u'_t,v'_t)_t=\Endpoints(s_t,\{ (u^i_{j_i}, v^i_{j_i}) \}_{i\in[\ell],j_i\in[k_{i}]})$ for each $t\in [k_0]$\\
\lIf{$A_0[(u'_1,v'_1),\dots, (u'_{k_0},v'_{k_0})]\not \in \Done$}{
Insert $A_0[(u'_1,v'_1),\dots, (u'_{k_0},v'_{k_0})]$ in $\Worklist$ and $\Done$
}
}
}
}
{\Return{$\{ (u,v)\colon \StartSymbol[(u,v)]\in \Done \}$}\label{line:algo_return}
}
\end{algorithm}

\Paragraph{Algorithm.}
The algorithm for $\drMCFL{d}{r}$-reachability is saturation-based, using a worklist $\Worklist$ that contains elements of the form $A[(u_1,v_1),\dots, (u_k,v_k)]$, where $k$ is the arity of the nonterminal $A$.
The algorithm maintains the invariant that such an element is inserted in $\Worklist$ iff there exist paths $P_i\colon u_i\Path v_i$ in $G$ such that $A\Derives(\Label(P_1),\dots, \Label(P_k))$.
Initially, $\Worklist$ is populated with all the basic (Type~1) rules of $\Grammar$ that are witnessed by the edges of $G$.
Then, while $\Worklist$ is not empty, an element $B[(u_1,v_1),\dots, (u_k,v_k)]$ is extracted.
The algorithm iterates over all production rules of $\Grammar$ that have $B$ on the right hand side, and tries to produce the left hand side of the rule.
For Type~2 and Type~3 rules, this is achieved by expanding one of the paths $P_i$ by a single node.
For Type~4 rules, this is achieved by choosing an independent edge of $G$ in order to start a new path.
Finally, Type~5 rules perform string concatenation, which is achieved by concatenating  paths that share endpoints according to $s_1\cdots s_k$.
See \cref{algo:mcfl_reach} for the formal description.

\Paragraph{Type-5 rules and path concatenation.}
The most complex production rules to handle are those of Type~5, which concern path concatenation.
Here we set up some convenient notation that is used in the algorithm.
Consider a set of pairs of nodes $X=\{ (u^i_{j_i}, v^i_{j_i})\mid {i\in [\ell], j_i\in [k_i]} \}$ for some positive integers $\ell$, $k_1,\dots, k_{\ell}$.
Given a string $s=x^{l_1}_{t_1}\cdots x^{l_h}_{t_h}$ with $l_i\in [\ell]$ and $t_i\in [k_{l_i}]$ (here $x$ is a fixed character representing a variable name used in a production rule of the grammar, not a variable representing a string),
we write $s\PathDerives X$ to denote that, for every $i\in[h-1]$ we have $v^{l_{i+1}}_{t_{i+1}} = u^{l_{i}}_{t_{i}}$.
If $s\PathDerives X$, we let $\Endpoints(s,X)=(u^{l_1}_{j_{t_1}}, v^{l_h}_{j_{t_h}})$.
Given a sequence $s_1,\dots, s_k$, we write $s_1,\dots, s_k\PathDerives X$ to denote that $s_i\PathDerives X$ for each $i\in[k]$.

\SubParagraph{Example.} We illustrate the above notation on an example, and explain how it is used in the algorithm. 
Consider a Type-5 rule as follows.

{\centering
  $ \displaystyle
    \begin{aligned} 
       A_0(s_1, s_2) &\gets A_1(x_1^1,x_2^1, x_3^1), A_2(x_1^2,x_2^2)\\
    \end{aligned}
  $ 
\par}
where $s_1=x_1^1x_2^1$ and $s_2=x_1^2x_2^2x_3^1$. 
Assume that the algorithm has already derived the predicates $A_1[(5,2),(2,3),(7,2)]$ and $A_2[(1,4),(4,7)]$,
and now it seeks to apply the above rule to derive a predicate $A_0[(u_1,v_1),(u_2,v_2)]$.
Since $s_1$ concatenates the strings represented by $x_1^1$ and $x_2^1$, the algorithm has to verify that this is indeed possible:~the end node of the subpath represented by $x_1^1$ is the start node of the subpath represented by $x_2^1$.
Similarly for the concatenation of the three subpaths in $s_2$.
Using our notation above, we have $X=\{(5,2),(2,3),(7,2),(1,4),(4,7)\}$ where the $i$-th pair of nodes marks the endpoints of the subpath represented by the corresponding variable $x_{t_i}^{l_i}$.
The notation $s\PathDerives X$ means that the subpaths represented in $X$ can indeed be joined to form $s$, because their endpoints match. 
In our example, $s_1=x_1^1x_2^1$, so the algorithm must check if the sequence of pairs of nodes $\{(5,2),(2,3)\}$ can be joined. 
This represents checking whether two paths $P_1\colon 5\Path 2$ and $P_2\colon 2\Path 3$ can be taken in sequence. 
This is indeed the case, and the algorithm forms a path $P_3\colon 5\Path 3$. 
The pair $(5,3)$ is thus in $\Endpoints(s_1,X)$.
Similarly, $s_2=x_1^2x_2^2x_3^1$, so we have the sequence $\{(1,4),(4,7),(7,2)\}$ from $X$, with $\Endpoints(s_2,X)=(1,2)$. 
This sequence is also such that $s_2\PathDerives X$, thus $s_1,s_2\PathDerives X$. 
Having verified this, the algorithm derives the new predicate $A_0[\Endpoints(s_1,X),\Endpoints(s_2,X)]=A_0[(5,3),(1,2)]$.

\subsection{Example}\label{SUBSEC:EXAMPLE}

Consider the following $\drMCFG{2}{2}$ $\Grammar'$, in normal form, for the language $\Language=\{w_1w_2\#w_2w_1\mid w_1,w_2\in\{0,1\}^*\}$, as already presented in \cref{tab:normal_form_example}.

\begin{minipage}{.35\linewidth}
  \centering
    \begin{align}
    &A_1(\epsilon)  \label{rule:ex_a1} \\
    &A(x, \epsilon) \gets A_1(x) \label{rule:ex_aa1}\\
    &A_2(x_1,x_20) \gets A(x_1, x_2) \label{rule:ex_a2a}\\
    &A(x_10,x_2) \gets A_2(x_1, x_2) \label{rule:ex_aa2}
    \end{align}
\end{minipage}
\hspace{1.5cm}
\begin{minipage}{.45\linewidth}
  \centering
    \begin{align}
&A_3(x_1,x_21) \gets A(x_1, x_2)  \label{rule:ex_a3a} \\
&A(x_11,x_2) \gets A_3(x_1, x_2) \label{rule:ex_aa3}\\
&A_4(x_1\#,x_2) \gets A(x_1, x_2)  \label{rule:ex_a4a} \\
&S(x_1y_1y_2x_2) \gets  A(x_1,x_2), A_4(y_1, y_2) \label{rule:ex_s}
    \end{align}
\end{minipage}

\begin{figure}[!h] 
\begin{subfigure}[t]{0.45\textwidth}
\centering
\begin{tikzpicture}[ line width=1pt, main/.style = {draw, rectangle,inner sep=4pt},scale=0.8] 
    \node[main] (a) at (0,0){\footnotesize a}; 
    \node[main] (b) at (1.5,0){\footnotesize b}; 
    \node[main] (c) at (1.5,1.5){\footnotesize c}; 
    \node[main] (d) at (3,1.5){\footnotesize d}; 
    \node[main] (e) at (4.5,1.5){\footnotesize e}; 

    \node[main,white] () at (4.5,-1){}; 

    \draw [->] (a) -- (b) node [midway, below=2pt] {\footnotesize $\pcOne{0}$}; 

    \draw[->] (b) edge[bend right] node[midway, right=2pt] {\footnotesize $\pcOne{1}$} (c);
    \draw [->] (c) edge[bend right] node [midway, left=2pt] {\footnotesize $\pcOne{\#}$} (b);
    \draw [->] (c) -- (d) node [midway, above=2pt] {\footnotesize $\pcOne{\epsilon}$}; 
    \draw [->] (d) -- (e) node [midway, above=2pt] {\footnotesize $\pcOne{0}$}; 
\end{tikzpicture}
\subcaption{\label{subfig:algo_ex_graph} Graph $G$, used to exemplify \cref{algo:mcfl_reach}.
}
\end{subfigure}
\qquad
\begin{subfigure}[t]{.45\linewidth}
\begin{tikzpicture}[ line width=1pt, main/.style = {draw, rectangle,inner sep=4pt},scale=0.8]
 \node[main] (a) at (0,0){\footnotesize a}; 
    \node[main] (b) at (1.8,0){\footnotesize b}; 
    \node[main] (c) at (1.8,1.8){\footnotesize c}; 
    \node[main] (d) at (3.6,1.8){\footnotesize d}; 
    \node[main] (e) at (5.4,1.8){\footnotesize e}; 

    \draw [->] (a) -- (b) node [midway, below=2pt] {\footnotesize $(1,\pcOne{0}$)}; 
    \draw [->] (b) edge[bend right] node [midway, right=0pt] {\footnotesize $(2,\pcOne{1})$}(c); 
    \draw [->] (c) edge[bend right] node [midway, left=0pt] {\footnotesize $(3,\pcOne{\#})$}(b);
    \draw [->] (c) -- (d) node [midway, above=2pt] {\footnotesize $(4,\pcOne{\epsilon})$}; 
    \draw [->] (d) -- (e) node [midway, above=2pt] {\footnotesize $(5,\pcOne{0})$}; 
    \path (a) edge [loop left] node {\footnotesize $(6,\pcOne{\epsilon})$} (a);
    \path (b) edge [loop right] node {\footnotesize $(7,\pcOne{\epsilon})$} (b);
    \path (c) edge [loop above] node {\footnotesize $(8,\pcOne{\epsilon})$} (c);
    \path (d) edge [loop above] node {\footnotesize $(9,\pcOne{\epsilon})$} (d);
    \path (e) edge [loop above] node {\footnotesize $(10,\pcOne{\epsilon})$} (e); 
            
\end{tikzpicture}
\subcaption{\label{subfig:algo_ex_graph_more} Graph $G$, with labeled edges and added self loops with label $\epsilon$ in every vertex.
}
\end{subfigure}
\caption{
Graphs used to exemplify the execution of \cref{algo:mcfl_reach}.
}
\label{fig:example_algo}
\end{figure}

We execute \cref{algo:mcfl_reach} on the graph shown in \cref{subfig:algo_ex_graph} using $\Grammar'$,
and illustrate how the algorithm determines the reachability of node $e$ from node $a$. 
Since many elements are added to the worklist $\Worklist$ during execution, 
we only outline the relevant actions, while omitting those that do not contribute to the sought reachability.
The relevant steps of the execution are as follows.

The algorithm starts with \cref{line:epsilon_edges}, adding self-loops labeled $\epsilon$ to each node, so the graph from \cref{subfig:algo_ex_graph} becomes the graph shown in \cref{subfig:algo_ex_graph_more}. 
Edges have also been labeled with numbers for easier reference.
Then the algorithm proceeds to \cref{line:algo_type1} for processing Type~1 rules.
It processes Rule~\ref{rule:ex_a1} and, from edges 6 and 7, respectively, adds elements $A_1[(a,a)]$ and $A_1[(b,b)]$ to $\Worklist$.
Now the algorithm enters the main loop in \cref{line:algo_main_loop}, and repeatedly
extracts elements from $\Worklist$ in \cref{line:algo_extraction}, adding other elements to it, until we reach $\StartSymbol[(a,e)]$, signifying that $a$ reaches $e$. 
In detail, we show which elements are added when each element $B$ is extracted from $\Worklist$ on \cref{line:algo_extraction}.

\begin{center}
\begin{tabular}{c|p{10cm}}
    \textbf{Element} $B$ & \textbf{Relevant actions}\\
    \hline
    $A_1[(a,a)]$ & In \cref{line:algo_type4} for Rule~\ref{rule:ex_aa1} and edge 8, adds $A[(a,a),(c,c)]$ to $W$.\\
    \hline
    $A_1[(b,b)]$ &  In \cref{line:algo_type4} for Rule~\ref{rule:ex_aa1} and edge 7, adds $A[(b,b),(b,b)]$ to $W$.\\
    \hline
    $A[(a,a),(c,c)]$ & In \cref{line:for_expand_epsilon}, expanding the derivation along the $\epsilon$-edge 4, adds $A[(a,a),(c,d)]$ to $\Worklist$.\\
    \hline
    $A[(a,a),(c,d)]$ & In \cref{line:algo_type3} from Rule~\ref{rule:ex_a2a} and edge 5, adds $A_2[(a,a),(c,e)]$ to $\Worklist$.\\
    \hline
    $A_2[(a,a),(c,e)]$ & In \cref{line:algo_type3} from Rule~\ref{rule:ex_aa2} and edge 1, adds $A[(a,b),(c,e)]$ to $\Worklist$.\\
    \hline
    $A[(b,b),(b,b)]$ & In \cref{line:algo_type3} from Rule~\ref{rule:ex_a3a} and edge 2, adds $A_3[(b,b),(b,c)]$ to $\Worklist$.\\
    \hline
    $A_3[(b,b),(b,c)]$ & In \cref{line:algo_type3} from Rule~\ref{rule:ex_aa3} and edge 2, adds $A[(b,c),(b,c)]$ to $\Worklist$.\\
    \hline
    $A[(b,c),(b,c)]$ & In \cref{line:algo_type3} from Rule~\ref{rule:ex_a4a} and edge 3, adds $A_4[(b,b),(b,c)]$ to $\Worklist$.
    
\end{tabular}
\end{center}

Finally, the element $A[(a,b),(c,e)]$ is extracted from $\Worklist$. 
In \cref{line:algo_type5} for Rule~\ref{rule:ex_s}, the algorithm considers the derivation $A_4[(b,b),(b,c)]$. 
It checks whether the elements of $A[(a,b),(c,e)]$ and $A_4[(b,b),(b,c)]$ can be concatenated in the order defined in Rule~\ref{rule:ex_s}. 
This would form the sequence of pairs of nodes $\{(a,b),(b,b),(b,c),(c,e)\}$. 
Notice that for every two adjacent pairs $(u_1,v_1), (u_2,v_2)$, we have $v_1=u_2$. 
This means that the paths represented by these pairs can be taken in sequence, thus the derivations follow the properties shown in the loop of \cref{line:algo_type5} and $\StartSymbol[(a,e)]$ can be inserted to $W$.
This signifies that $a$ reaches $e$, as per \cref{line:algo_return}.
\subsection{Correctness and Complexity}\label{SUBSEC:CORRECTNESS_COMPLEXITY}

In the following we discuss the completeness and soundness of \cref{algo:mcfl_reach}, and conclude \cref{thm:upper_bound} with the complexity analysis.
We refer to \cref{SEC:APP_MCFG_REACHABILITY} for the formal proofs.

First, \cref{lem:soundness} states the invariant for soundness.
It follows by induction on the iterations of the main loop, proving that the algorithm correctly derives new predicates from existing ones.

\begin{restatable}{lemma}{lemsoundness}\label{lem:soundness}
For every element $A[(u_1, v_1),\dots, (u_k,v_k)]$ inserted in $\Worklist$,
there exist paths $\{P_i\colon u_i\Path v_i\}_{i\in [k]}$ such that $A\Derives (\Label(P_1), \dots, \Label(P_k))$.
\end{restatable}

Second, \cref{lem:completeness} establishes completeness.
The proof follows by induction, which generalizes the analogous proof for the standard algorithm for CFL reachability to arbitrary dimension and rank.
Intuitively, now a predicate $A$ might summarize multiple independent paths that are eventually merged using Type-5 rules (\cref{line:algo_type5}), as opposed to a single path which is the case for CFL reachability.

\begin{restatable}{lemma}{lemcompleteness}\label{lem:completeness}
For every non-terminal $A$ of arity $k$ and every set of paths $\{P_i\colon u_i\Path v_i\}_{i\in[k]}$ such that $A\Derives(\Label(P_1), \dots, \Label(P_k))$,
the element $A[(u_1, v_1),\dots, (u_k,v_k)]$ is inserted in $\Worklist$.
\end{restatable}

The correctness now follows immediately from \cref{lem:soundness} and \cref{lem:completeness}.
Indeed, the two lemmas guarantee that for every two nodes $u$ and $v$,
 the element $\StartSymbol([(u,v)])$ is inserted in (and thus extracted from) $\Worklist$ iff there is a path $P\colon u\Path v$ such that $\StartSymbol\Derives \Label(P)$, and thus iff $v$ is $\Language(\Grammar)$-reachable from $u$.
 
We now turn our attention to the complexity of \cref{algo:mcfl_reach}.
Intuitively, the normal form of \cref{SUBSEC:NORMAL_FORM} has the following consequences for the runtime.

First, consider a non-terminal $B[(u_1, v_1),\dots, (u_k,v_k)]$ being extracted from the worklist.
Since the grammar $\Grammar$ has dimension $d$, we have $k\leq d$, meaning we have at most $n^{2d}$ such extractions for each non-terminal $B$.
The algorithm attempts to extend $B[(u_1, v_1),\dots, (u_k,v_k)]$ by a single terminal symbol according to a rule of Type~2-Type 4, by iterating over all nodes $u$ with an in-edge to or out-edge from $u_i$ or $v_i$.
This incurs an increase in running time by a factor of $O(\delta)$, where $\delta$ is the maximum degree of $G$, resulting in a contribution to the complexity of the order $O(\delta\cdot n^{2d})$.

Second, each rule of Type 5 relates at most $n^{d(r+1)}$ nodes.
To observe this, recall that, since the rank of $\Grammar$ is $r$, each Type 5 rule has $\leq r+1$ predicates (one on the left hand side and $\leq r$ on the right hand side).
Since the dimension of $\Grammar$ is $d$, each predicate relates $2d$ nodes (i.e., the predicate summarizes $d$ sub-paths, each marked by its two endpoints).
However, not all $2d(r+1)$ nodes are free, since sub-paths that are merging on the left-hand side must have matching endpoints.
This leaves only $d(r+1)$ nodes free, resulting in $n^{d(r+1)}$ possible choices for each given non-terminal symbol.
Formally we have the following lemma, which concludes the proof of \cref{thm:upper_bound}.

\begin{restatable}{lemma}{lemcomplexity}\label{lem:complexity}
Given a $\drMCFG{d}{r}$ $\Grammar$ and a graph $G$ of $n$ nodes and maximum degree $\delta$,
\cref{algo:mcfl_reach}  takes
\begin{compactenum}
\item \label{item:complexity_r1} $O(\Poly(|\Grammar|)\cdot \delta\cdot n^{2d})$ time, if $r=1$, and
\item \label{item:complexity_rlarge} $O(\Poly(|\Grammar|)\cdot n^{d(r+1)})$ time, if $r>1$.
\end{compactenum}
\end{restatable}

\section{Lower Bounds}\label{SEC:LOWER_BOUNDS}

A natural follow-up question to \cref{thm:upper_bound} is whether the MCFL reachability and membership problems admit faster algorithms.
In this section we turn our attention to lower bounds, and prove \cref{thm:ov_hard} and \cref{thm:triangle_hard}.

\subsection{SETH Hardness of $\drMCFL{d}{1}$ Membership via Orthogonal Vectors}\label{subsec:seth_hardness}

We focus on the membership problem for $\drMCFGs{d}{1}$, and establish a conditional lower bound of $n^{2d}$ (wrt polynomial improvements) using the popular Orthogonal Vectors Hypothesis (OVH).
Conceptually, a $d$-dimensional MCFG parses $d$ substrings simultaneously. 
In each step, each substring can be extended in $2$ directions, to the left and to the right, each yielding $n$ choices for the substring to extend with. 
This contributes to the complexity a factor of $n^2$ per dimension, totaling $n^{2d}$. 
The following result makes this intuition formal by proving this cost is unavoidable in general.

\Paragraph{The orthogonal vectors problem.}
The $\NumVectorSets$-orthogonal vectors ($k$-OV) problem, for some integer $\NumVectorSets\geq 2$, is defined as follows.
Given $\NumVectorSets$ sets $X_1,\dots, X_{\NumVectorSets}\subseteq 2^{\{0,1\}^{\NumBits}}$ of Boolean $\NumBits$-dimensional vectors $X_i=\{ x^{i}_{j} \}_{j \in [\NumVectors]}$,
the task is to determine whether there exist $\{j_i\}_i$ such that the vectors $x^{i}_{j_i}\in X_i$ are orthogonal, i.e.,
\begin{align}
\sum_{\ell=1}^{\NumBits} \prod_{i=1}^k x^{i}_{j_i}[\ell]=0
\label{eq:orthogonality}
\end{align}
In other words, in every coordinate $\ell$, at least one of the $x^i_{j_i}$ has value $0$.
A simple algorithm to solve orthogonal vectors is to iterate over all $\NumVectors^{\NumVectorSets}$ tuples of vectors in the product $\bigtimes_iX_i$ and check whether each tuple is orthogonal, which yields a straightforward bound $O(\NumVectors^{\NumVectorSets}\cdot \NumBits)$. 
The corresponding hypothesis (OVH) states that, for any fixed $k$, the problem is not solvable in $O(\NumVectors^{\NumVectorSets-\epsilon}\cdot \Poly(\NumBits))$ time, for any fixed $\epsilon>0$~\cite{Williams19}.
It is also known that the Strong Exponential Time Hypothesis (SETH) implies the Orthogonal Vectors Hypothesis~\cite{Williams05}.
Here we prove the following conditional lower bound for $\drMCFL{d}{1}$ membership (and thus also reachability).

\thmovhard*

\cref{thm:ov_hard} relies on the observation that the orthogonality between $\NumVectorSets$ vectors (\cref{eq:orthogonality}), for $\NumVectorSets$ even, can be expressed as a $\drMCFG{\NumVectorSets/2}{1}$, i.e., of dimension that is half the number of vector sets $\NumVectorSets$.

\Paragraph{Outline.}
First, given positive integers $\NumBits$, $\NumVectors$, and $\NumVectorSets$, where $\NumVectorSets$ is even, we define a template language $\TemplateLanguage_{\NumVectorSets, \NumVectors}^{\NumBits}$, which encodes instances of the orthogonal vectors problem as strings.
Then, given a positive even integer $\NumVectorSets$, we define a $\drMCFG{\NumVectorSets/2}{1}$ $\Grammar_{\NumVectorSets}$ with the guarantee that every string $w\in \TemplateLanguage_{\NumVectorSets, \NumVectors}^{\NumBits}\cap \Language(\Grammar_{\NumVectorSets})$ encodes an orthogonal set of vectors.
Finally, to decide  whether an instance of $\NumVectorSets$-OV $X_1,\dots, X_{\NumVectorSets}$ contains an orthogonal set, we simply encode it as a string $w\in \TemplateLanguage_{\NumVectorSets, \NumVectors}^{\NumBits}$, spending $O(|w|)$ time, and then ask whether $w\in \Language(\Grammar_{\NumVectorSets})$.
Hence, the time to decide orthogonality is upper-bounded by the time required to check whether $w\in \Language(\Grammar_{\NumVectorSets})$, leading to the lower bound of \cref{thm:ov_hard}.
In the following we make these insights formal.

\Paragraph{The template language.}
Given the vector sets $X_1,\dots, X_{\NumVectorSets}$, one natural way to encode them as a string is to juxtapose all vectors $x^i_j$ next to each other, separated by appropriate markers.
However, to guarantee that $w\in \TemplateLanguage_{\NumVectorSets, \NumVectors}^{\NumBits}\cap \Language_{\NumVectorSets}$ iff there is an orthogonal set, this choice appears to require that $\Language_{\NumVectorSets}$ be $\drMCFL{\NumVectorSets}{1}$, as opposed to a $\drMCFL{\NumVectorSets/2}{1}$, i.e., we need to use twice as large a dimension.
To achieve the desired dimension $\NumVectorSets/2$, we instead reverse the vectors of each odd set (i.e., sets $W_{2i+1}$).
Formally, we define the language
\begin{align*}
\TemplateLanguage_{\NumVectorSets, \NumVectors}^{\NumBits} = &\Bigg \{ \#_1w^1\Delimiter_1^2w^2w^3\Delimiter_3^4w^4\dots w^{\NumVectorSets-1}\Delimiter_{\NumVectorSets-1}^{\NumVectorSets}w^{\NumVectorSets} &\text{ such that }  \\
&  \forall i\in[\NumVectorSets] \text{ with } i \text{ odd, we have } w^i=w^i_1\#_iw^i_2\#_i\cdots \#_iw^i_{\NumVectors} \#_i  &\text{ and }   \\
&  \forall i\in[\NumVectorSets] \text{ with } i \text{ even, we have } w^i=\#_iw^i_1\#_iw^i_2\cdots \#_iw^i_{\NumVectors}  &\text{ and }   \\
& \forall i\in[\NumVectorSets] \forall j\in [\NumVectors], \text{ we have } w^{i}_{j} \in \{0,1\}^{\NumBits} \Bigg \}
\numberthis\label{eq:TemplateLanguage}
\end{align*}
Given a $\NumVectorSets$-OV instance $X_1,\dots, X_{\NumVectorSets}$, we construct a string $w\in \TemplateLanguage_{\NumVectorSets, \NumVectors}^{\NumBits}$ by taking $w^i_{j}$ to be the $j$-th vector of $X_i$, if $i$ is even, and the reverse of the $j$-th vector of $X_i$, if $i$ is odd.
For example, consider the $2$-OV instance $X_1=\{110, 010\}$ and $X_2=\{011, 101\}$.
We construct the string
\[
w= \#_1 011\#_1 010 \#_1 \Delimiter_{1}^{2} \#_2 011 \#_2 101 \quad \in \TemplateLanguage_{\NumVectorSets, \NumVectors}^{\NumBits}
\numberthis\label{eq:OVStringExample}
\]

\Paragraph{The orthogonal-vectors MCFG.}
We present a $\drMCFG{k/2}{1}$ $\Grammar_{\NumVectorSets}$ such that $\Language(\Grammar_{\NumVectorSets})$ parses the strings $w\in \TemplateLanguage_{\NumVectorSets, \NumVectors}^{\NumBits}$ that contain an orthogonal set of vectors.
For simplicity of notation, let $d=k/2$.
The grammar consists of the following rules.
\begin{compactenum}
\item 
The basic rule initiates a matching on the boundaries of the substrings encoding the different vector sets.
\begin{align}
A(\Delimiter_1^2,\dots, \Delimiter_{\NumVectorSets-1}^{\NumVectorSets})\label{eq:ov-basic}
\end{align}

\item 
A production rule further grows each substring of predicate $A$ in both directions non-deterministically, in a search to identify the positions that mark the beginnings of $\NumVectorSets$ orthogonal vectors.
\begin{align}
A(s_1x_1t_1,\dots, s_dx_dt_d)\gets  A(x_1,\dots,x_d) \quad  \text{for} \quad s_i,t_i\in\{\epsilon,0,1,\#_1,\dots, \#_k\}
\label{eq:g2}
\end{align}

\item 
The predicate $B$ starts matching the orthogonal vectors on the positions marked by $A$.
Notice that vectors in odd sets grow right-to-left (as they have been reversed), while vectors in even sets grow left-to-right.
\begin{align}
B(s_1\#_1x_1\#_2t_1,\dots, s_{d}\#_{k-1}x_d\#_k t_d) &\gets A(x_1,\dots, x_d)\label{eq:g3}\\ &\text{for } s_i,t_i\in\{0,1\} \text{ and } \left(\exists i \in [d], s_i=0 \text{ or } t_i=0\right)\notag
\end{align}

\item
A production rule extends the existing orthogonal vectors to the next coordinate, provided that the extension maintains orthogonality.
Similarly to Rule~\ref{eq:g3}, this rule extends each $x_i$ with pairs of terminals in $\{0,1\}$.
\begin{align}
B(s_1x_1t_1,\dots, s_dx_dt_d)&\gets  B(x_1,\dots,x_d) \text{ for } s_i,t_i\in\{0,1\} \text{ and } \left(\exists i \in [d], s_i=0 \text{ or } t_i=0\right)\label{eq:g4}
\end{align}
\item

The predicate $C$ signifies that $\NumVectorSets$ orthogonal vectors have been found.
This is captured by reaching the marker $\#_1$ on the vector of the first set.
\begin{align}
&C(\#_1x_1,\dots, x_d)\gets B(x_1,\dots, x_{d})
\label{eq:g5}
\end{align}

\item 
A production rule further extends $C$ arbitrarily to cover the whole input string.
\begin{align}
&C(s_1x_1t_1,\dots, s_dx_dt_d)\gets C(x_1,\dots, x_d) \quad \text{for} \quad s_i,t_i\in \{\epsilon, 0,1, \#_1,\dots, \#_{k} \}
\label{eq:g6}
\end{align}
\item Finally, concatenating the substrings derivable by $C$ yields the start symbol.
\begin{align}
&\StartSymbol(x_1\dots x_d)\gets C(x_1,\dots,x_d)
\label{eq:g7}
\end{align}
\end{compactenum}

\SubParagraph{Example.}
\cref{fig:ov_example} shows a derivation for the string $w$ in \cref{eq:OVStringExample}, identifying the orthogonal vectors $010\in X_1$ and $101\in X_2$. Predicate $A$ produces some a part of the vector sets, $B$ produces the orthogonal vectors, and $C$ produces the remainder of the vector sets.

\begin{figure}
\centering
\begin{tikzpicture}[node distance={15mm}, line width=1pt, main/.style = {draw, rectangle},scale=0.8] 
    \node[main] at (-1.3,1.2) (eps) [white]{}; 
    \node[main] at (0,1.2) (a1) [black, text=black]{\footnotesize$A(\Delimiter_{1}^{2})$}; 
    \node[main] at (0,0) (a2) [black, text=black]{\footnotesize$A(\Delimiter_{1}^{2} \#_2 011)$}; 
    \node[main] at (3.3,1.2) (b1) [black, text=black]{\footnotesize$B(0 \#_1 \Delimiter_{1}^{2} \#_2 011 \#_2 1)$}; 
    \node[main] at (3.3,0) (b2) [black, text=black]{\footnotesize$B(010 \#_1 \Delimiter_{1}^{2} \#_2 011 \#_2 101)$}; 
    \node[main] at (8,1.2) (c1) [black, text=black]{\footnotesize$C(\#_1 010 \#_1 \Delimiter_{1}^{2} \#_2 011 \#_2 101)$}; 
    \node[main] at (8,0) (c2) [black, text=black]{\footnotesize$C(\#_1 011\#_1 010 \#_1 \Delimiter_{1}^{2} \#_2 011 \#_2 101)$}; 
    \node[main] at (13.5,1.2) (s) [black, text=black]{\footnotesize$\StartSymbol(\#_1 011\#_1 010 \#_1 \Delimiter_{1}^{2} \#_2 011 \#_2 101)$};  
    \draw [->] (eps) -- (a1)  node [midway,above] {\footnotesize\ref{eq:ov-basic}};
    \draw [->] (a1) -- (a2)  node [midway,right] {\footnotesize\ref{eq:g2}+};
    \draw [->] (a2) -- (1.2,0) -- (1.2,1.2) --  (b1)  node [midway,above] {\footnotesize \ref{eq:g3}};
    \draw [->] (b1) -- (b2)  node [midway,right] {\footnotesize\ref{eq:g4}+};
    \draw [->] (b2) -- (5.3,0) -- (5.3,1.2) --  (c1)  node [midway,above] {\footnotesize \ref{eq:g5}};
    \draw [->] (c1) -- (c2)  node [midway,right] {\footnotesize\ref{eq:g6}+};
    \draw [->] (c2) -- (10.6,0) -- (10.6,1.2) --  (s)  node [midway,above] {\footnotesize \ref{eq:g7}};
\end{tikzpicture}
\caption{
Derivation of the string $w=\#_1 011\#_1 010 \#_1 \Delimiter_{1}^{2} \#_2 011 \#_2 101$, encoding a $2$-OV instance for $\NumVectorSets=2$. 
An edge labeled $R+$ indicates multiple applications of rule $R$. 
}
\label{fig:ov_example}
\end{figure}


\subsection{BMMH Hardness of $\drMCFL{1}{1}$ Reachability via Triangle-Freeness}\label{subsec:triangle_hardness}

We now turn our attention to $\drMCFL{1}{1}$ reachability, and establish a conditional lower bound of $n^{3}$ based on the classic problem of Boolean Matrix Multiplication (BMM).
Intuitively, in order to deduce the MCFL-reachability of a node $v$ from another node $u$, the algorithm might have to identify up to $\Omega(n^2)$ reachable node pairs, suffering a cost of $\Omega(n)$ for each such pair.
Our lower bound makes this insight formal.

In algorithmic theory parlance, there is a distinction between algebraic and combinatorial algorithms for performing BMM.
Although the separation between the two domains is not entirely formal, generally speaking, combinatorial algorithms do not depend on the properties of the semiring used for multiplication.
The corresponding BMM hypothesis (BMMH) states that there is no combinatorial algorithm for multiplying two $n\times n$ boolean matrices in $O(n^{3-\epsilon})$ time, for any fixed $\epsilon > 0$~\cite{Williams19}.
We show that the same combinatorial lower bound holds for $\drMCFL{1}{1}$ reachability under BMMH.

\thmtrianglehard*

\begin{figure}
\centering
\begin{tikzpicture}
[-> >=1pt,auto,node distance=1.5cm,scale = 0.8,transform shape,line width=1pt,inner sep=4pt,rectangle]
\tikzstyle{state}=[circle, inner sep=2pt, draw=black, line width=1pt, minimum size=7mm]


  \newcommand{\ystep}{0.5}
  \newcommand{\xstep}{1}

  	\begin{scope}[shift={(0,0)}]
      \node[state] (1) at (0*\xstep,5*\ystep) {\Large$1$};
      \node[state] (2) at (2*\xstep,3*\ystep) {\Large$2$};
      
      \node[state] (3) at (0*\xstep,3*\ystep) {\Large$3$};
      \node[state] (4) at (0*\xstep,1*\ystep) {\Large$4$};
      
      \draw[thick,-] (1) -- (2) node[midway,sloped,left, rotate=-90] {$ $};
      \draw[thick,-] (1) -- (3) node[midway,sloped,left, rotate=-90] {$ $};
      \draw[thick,-] (3) -- (2) node[midway,sloped,left, rotate=-90] {$ $};
      \draw[thick,-] (2) -- (4) node[midway,sloped,left, rotate=-90] {$ $};
    \end{scope}

    \begin{scope}[shift={(6,0)}]
    \node[main] (s) at (-2*\xstep,6*\ystep) {$u$};
    
    \node[main] (a1) at (0*\xstep,6*\ystep) {$u_1$};
    \node[main] (a2) at (0*\xstep,4*\ystep) {$u_2$};
    \node[main] (a3) at (0*\xstep,2*\ystep) {$u_3$};
    \node[main] (a4) at (0*\xstep,0*\ystep) {$u_4$};
    
    \node[main] (b1) at (2*\xstep,6*\ystep) {$v_1$};
    \node[main] (b2) at (2*\xstep,4*\ystep) {$v_2$};
    \node[main] (b3) at (2*\xstep,2*\ystep) {$v_3$};
    \node[main] (b4) at (2*\xstep,0*\ystep) {$v_4$};
    
    \node[main] (c1) at (4*\xstep,6*\ystep) {$y_1$};
    \node[main] (c2) at (4*\xstep,4*\ystep) {$y_2$};
    \node[main] (c3) at (4*\xstep,2*\ystep) {$y_3$};
    \node[main] (c4) at (4*\xstep,0*\ystep) {$y_4$};
    
    \node[main] (d1) at (6*\xstep,6*\ystep) {$z_1$};
    \node[main] (d2) at (6*\xstep,4*\ystep) {$z_2$};
    \node[main] (d3) at (6*\xstep,2*\ystep) {$z_3$};
    \node[main] (d4) at (6*\xstep,0*\ystep) {$z_4$};
    
    \node[main] (e) at (8*\xstep,6*\ystep) {$v$};

    
    \draw[->, thick] (s) to node[midway, above]{$0$} (a1) ;
    \draw[->, thick] (a1) to node[midway, left] {$0$}  (a2);
    \draw[->, thick] (a2) to node[midway, left] {$0$} (a3) ;
    \draw[->, thick] (a3) to node[midway, left] {$0$} (a4);
    
    \draw[->, thick] (a1) -- (b2);
    \draw[->, thick] (a1) -- (b3);
    
    \draw[->, thick] (a2) -- (b1);
    \draw[->, thick] (a2) -- (b3);
    \draw[->, thick] (a2) -- (b4);
    
    \draw[->, thick] (a3) -- (b1);
    \draw[->, thick] (a3) -- (b2);
    
    \draw[->, thick] (a4) -- (b2);
    
    \draw[->, thick] (b1) -- (c2);
    \draw[->, thick] (b1) -- (c3);
    
    \draw[->, thick] (b2) -- (c1);
    \draw[->, thick] (b2) -- (c3);
    \draw[->, thick] (b2) -- (c4);
    
    \draw[->, thick] (b3) -- (c1);
    \draw[->, thick] (b3) -- (c2);
    
    \draw[->, thick] (b4) -- (c2);
    
    \draw[->, thick] (c1) -- (d2);
    \draw[->, thick] (c1) -- (d3);
    
    \draw[->, thick] (c2) -- (d1);
    \draw[->, thick] (c2) -- (d3);
    \draw[->, thick] (c2) -- (d4);
    
    \draw[->, thick] (c3) -- (d1);
    \draw[->, thick] (c3) -- (d2);
    
    \draw[->, thick] (c4) -- (d2);

    \draw[->, thick] (d4) -- (d3) node[midway,right] {$1$};
    \draw[->, thick] (d3) -- (d2) node[midway,right] {$1$};
    \draw[->, thick] (d2) -- (d1) node[midway,right] {$1$};
    \draw[->, thick] (d1) -- (e) node[midway,above] {$1$};

     \end{scope}

\end{tikzpicture}
\caption{
Reduction of triangle-freeness on a graph $G'$ (left) to $\TriangleL$-reachability from $u$ to $v$ on a graph $G$ (right).
Non-labeled edges correspond to $\epsilon$-labeled edges in the reduction.
}
\label{fig:triangle}
\end{figure}

Note that \cref{thm:triangle_hard} indeed relates a decision problem (single-pair $\drMCFL{1}{1}$-reachability) to a function problem (BMM).
The proof is via a reduction from the problem of triangle-freeness, i.e., given an undirected graph $G'$, determine whether $G'$ contains a triangle as a subgraph.
The problem is known to have no combinatorial subcubic algorithms under BMMH~\cite{Williams18b}.

\Paragraph{Reduction.}
We now present the reduction in detail.
First, we choose the alphabet $\Terminals=\{0,1\}$, and construct the $\drMCFL{1}{1}$ language $\TriangleL=\{0^k1^k\colon k\geq 1\}$,
given by the grammar with rules
(i)~$\StartSymbol(01)$, and
(ii)~$\StartSymbol(0x1)\gets \StartSymbol(x)$.

Next, we construct a graph $G$ given the input graph $G'$.
We assume wlog that the node set of $G'$ is $[n]$.
The node set of $G$ consists of two distinguished nodes $u$ and $v$, as well as four nodes $u_i,v_i, y_i, z_i$ for every $i\in[n]$.
The construction guarantees that $v$ is $\TriangleL$-reachable from $u$ iff $G'$ contains a triangle.
We have an edge $(u, 0, u_1)$ and $(v_1, 1, v)$, as well as edges $(u_i, 0, u_{i+1})$ and $(v_{i+1}, 1, v_i)$ for every $i\in[n-1]$.
Moreover, for every $i,j\in[n]$ such that $G'$ has an edge $(i,j)$, we have edges $(u_i, \epsilon, v_j)$, $(v_i, \epsilon, y_j)$, and $(y_i, \epsilon, z_j)$ in $G$.
See \cref{fig:triangle} for an illustration.

\section{Experiments}\label{SEC:EXPERIMENTS}

Here we report on an implementation of \cref{algo:mcfl_reach} for computing MCFL reachability, and an experimental evaluation of the algorithm wrt the MCFLs $\Language(\Grammar_d^+)$ that approximate interleaved Dyck reachability (see \cref{SEC:MODELING}), on a standard benchmark for taint analysis for Android. 

\subsection{Implementation}\label{SUBSEC:IMPLEMENTATION}

Our implementation is done in Go and closely follows the pseudocode in  \cref{algo:mcfl_reach}.
Besides a few standard algorithmic engineering choices (e.g., for testing efficiently the condition in \cref{{line:algo_type5_pathmatching}} of the algorithm), we employed two high-level optimizations, that we explain here in more detail.

\Paragraph{Restriction on potentially reachable pairs.} 
Since MCFGs are, wlog, non-permuting (see \cref{SEC:PRELIMINARIES}), when the algorithm generates an element of the form $A[(u_1, v_1),\dots, (u_k, v_k)]$, this can only eventually produce the start non-terminal $\StartSymbol[(u,v)]$ connecting two nodes $u$ and $v$ if each $v_i$ can reach (in the sense of plain graph reachability) $u_{i+1}$.
In our implementation, we pre-compute plain graph reachability between all pairs of nodes, and the algorithm never produces elements such as the above, when the reachability condition is not met.
This optimization leads to fewer elements being needlessly added to the worklist $\Worklist$, making the overall computation faster.
It is also generic, in the sense that it applies to all MCFGs.

\Paragraph{Cycle elimination.} 
Cycle elimination is a common trick for speeding up constraint-based static analyses.
In our setting of $\Language(\Grammar_d^+)$ reachability, it is realized as follows.
Assume that there are two nodes $u,v$ that have been deemed inter-reachable, i.e.,
there exists  $P_1\colon u\Path v$ with $\Label(P_1)\in \IDL$ and $P_2\colon v\Path u$ with $\Label(P_2)\in \IDL$. 
Then $u$ and $v$ can be contracted to a single node $u$, without affecting the reachability of other nodes.
Formally, this is implied by \cref{lem:insert_string_dyck}.
In our implementation, we contract such pairs of nodes that are deemed reachable with respect to $\Language(\Grammar_1^+)$, for which reachability can be obtained efficiently (and which also implies $\IDL$ reachability).

\subsection{Evaluation}\label{SUBSEC:EVALUATION}

We now present the evaluation of our tool, as described in \cref{SUBSEC:IMPLEMENTATION}.

\Paragraph{Choice of benchmarks.}
Our benchmark set is a taint analysis for Android~\cite{Huang2015}, which is phrased as interleaved Dyck reachability on the underlying dataflow graphs.
This is a standard and challenging benchmark, used to evaluate earlier work on interleaved Dyck reachability~\cite{Zhang2017,Ding2023,Conrado2024}. Another standard benchmark has been a value-flow analysis~\cite{Ding2023}, but overapproximation algorithms already coincide with results obtained by CFL reachability (i.e., MCFL reachability in dimension 1)~\cite{Conrado2024}. 

\Paragraph{Setup.}
Our experiments were run on a conventional laptop MacBook Air with an Apple M1 chip of 3.2 GHz and 8GB of RAM.
We set a time out (TO) of 4 hours.

\Paragraph{Overapproximation.} 
To our knowledge, all prior work on interleaved Dyck reachability focuses on overapproximations.
Due to this fact, the results we obtain (i.e., the number of node pairs reported as reachable) are not directly comparable to any other underapproximate baseline.
To get an indication of coverage, we instead compare our results with those obtained using the most recent overapproximate method~\cite{Conrado2024}. 
This is the most recent, and more precise, work, in a line of several approaches developed and evaluated on this dataset.

Because of the running time limitations of an on-demand approach in reporting all reachable pairs, we use our underapproximation ($\Grammar_2^+$) to speed up the running time of the overapproximation algorithm. 
Specifically, we only run the on-demand algorithm on node pairs that have not been confirmed reachable by the underapproximation, nor discarded by some general (not on-demand) overapproximation.  
This can greatly decrease the number of pairs the algorithm must consider. 
For example, for the \texttt{fakedaum} benchmark, using our underapproximation, the overapproximation algorithm runs 36 times faster.

At the end of this process,  what remains are node pairs that are in the ``gray zone'', i.e., they are neither provably reachable nor provably unreachable.
Naturally, when the two methods coincide, we obtain a sound and complete analysis.

%


\begin{table}
\caption{
The number of reachable pairs wrt interleaved Dyck reachability reported by the overapproximation, as well as the underapproximations provided by grammars $\Grammar_1^+$ and $\Grammar_2^+$.
Gray marks entries where the underapproximation matches the overapproximation.
Time refers to the total time taken for computing reachability wrt to the approximation, and, respectively, for grammars $\Grammar_1^+$ and $\Grammar_2^+$.
}
\pgfplotstableread[col sep=comma]{MCFLdata.csv}\datatable
\centering
\begin{filecontents*}{MCFLdata.csv}
Benchmark,Nodes,Edges,Overapprox.,TimeO,1dim,2dim,Time
backflash,544,2048,2625,2m 51.5s,2625,2625,1m 34.9s
batterydoc,1674,4790,3340,30m 15.5s,2804,2839,38.1s
droidkongfu,734,1983,3401,2h 18m 30.4s,2906,3361,28m 48s
fakebanker,434,1103,251,12.5s,249,251,3.9s
fakedaum,1144,2603,2282,3m 5.2s,1132,2276,50m 47.2s
faketaobao,222,450,59,8.2s,57,59,0.8s
jollyserv,488,998,164,6.2s,155,164,1.5s
loozfon,152,323,93,4.2s,76,93,3.6s
roidsec,553,2026,13052,3m 14s,12284,13052,11m 6.5s
uranai,568,1246,143,4.1s,143,143,1.6s
zertsecurity,281,710,794,9.6s,779,794,20.3s
scipiex,1809,5820,TO,TO,1699,2115,7m 16.6s
\end{filecontents*}
\pgfplotstabletypeset[%
columns={Benchmark,Nodes,Edges,Overapprox.,TimeO,1dim,2dim,Time},
col sep=comma,
assign column name/.style={/pgfplots/table/column name={\textbf{#1}}},
every head row/.style={before row=\toprule, after row=\midrule, font=bold},
every last row/.style={after row=\bottomrule},
display columns/0/.style={string type},
columns/{Benchmark}/.style={column type={l|},string type,
postproc cell content/.prefix code={%
\pgfkeyssetvalue{/pgfplots/table/@cell content}{ { \ttfamily ##1 }}
}
},
columns/{1dim}/.style={column type={r},string type , column name=$\Language(\Grammar_1^+)$,
        postproc cell content/.prefix code={%
            \pgfplotstablegetelem{\pgfplotstablerow}{Overapprox.}\of\datatable
            \ifthenelse{\equal{##1}{\pgfplotsretval}}
            {\pgfkeyssetvalue{/pgfplots/table/@cell content}{ \cellcolor{lightgray}{ \bfseries \num{##1} }}}
            {\pgfkeyssetvalue{/pgfplots/table/@cell content}{ \num{##1} }}
        }
},
columns/{2dim}/.style={column type={|r},string type, column name=$\Language(\Grammar_2^+)$,
        postproc cell content/.prefix code={%
            \pgfplotstablegetelem{\pgfplotstablerow}{Overapprox.}\of\datatable
            \ifthenelse{\equal{##1}{\pgfplotsretval}}
            {\pgfkeyssetvalue{/pgfplots/table/@cell content}{ \cellcolor{lightgray}{ \bfseries \num{##1} }}}
            {\pgfkeyssetvalue{/pgfplots/table/@cell content}{ \num{##1} }}
        }
},
every row no 10/.style={
    after row={\midrule}
},
columns/{Nodes}/.style={column type={r},},
columns/{Edges}/.style={column type={r},},
columns/{Overapprox.}/.style={column type={|r},string type, column name=Overapp.,
    postproc cell content/.prefix code={%
            \ifthenelse{\equal{\string ##1}{\string TO}}
            { \pgfkeyssetvalue{/pgfplots/table/@cell content} {##1}}
            {\pgfkeyssetvalue{/pgfplots/table/@cell content} {\num{##1}}}%
        }
},
columns/{TimeO}/.style={column type={r|},string type, column name=Time},
columns/{Time}/.style={column type={r}, string type},
]{MCFLdata.csv}
\label{table:results}
\end{table}

\Paragraph{Results.} 
We utilize the grammars $\Grammar_d^+$ (see \cref{SEC:MODELING}) to underapproximate interleaved Dyck reachability.  Our results are shown in \cref{table:results}.
We observe that reachability wrt $\Language(\Grammar_1^+)$ reports several reachable pairs, and this number is sometimes close to the overapproximation, becoming tight in two benchmarks (\texttt{backflash} and \texttt{uranai}).
On the other hand, $\Language(\Grammar_1^+)$ is also sometimes far from the overapproximation, such as in benchmarks \texttt{batterdoc}, \texttt{droidkongfu}, and \texttt{fakedaum}.
Moving to $\Language(\Grammar_2^+)$, we observe that it reports more reachable pairs than $\Language(\Grammar_1^+)$ in all benchmarks where $\Language(\Grammar_1^+)$ has not provably converged to the correct solution.
Remarkably, $\Language(\Grammar_2^+)$ matches the overapproximation in $8$ out of $11$ benchmarks.

On the remaining $3$ benchmarks, $\Language(\Grammar_2^+)$ is able to confirm $94.3\%$ of the paths reported by the overapproximation, on average.
To our knowledge, this is the first time that this taint analysis has been solved with such a high, provable precision.
The overapproximation times out on \texttt{scipiex}. 
Here we again observe that $\Language(\Grammar_2^+)$ reports significantly more reachable pairs than $\Language(\Grammar_1^+)$, highlighting the importance of context sensitive analysis. 
We also note that our weaker grammar $\Language(\Grammar_2^{\circ})$ (\cref{SUBSEC:BASIC_RULES}) only  matches the overapproximation in $3$ out of the $11$ benchmarks, which shows the importance of choosing an expressive grammar. Finally, we remark that three benchmarks from the original benchmark set: \textsc{phospy}, \textsc{simhosy}, and \textsc{skullkey}, were omitted from the analysis, as all reported algorithms timed out. These all contain graphs with over 4,000 vertices and 10,000 edges.

\Paragraph{Running time.} 
Although the worst-case complexity of our algorithm for $\Language(\Grammar^+_2)$ reachability is $O(n^6)$, 
we observe that it can run efficiently on graphs with thousands of nodes.
Naturally, some graphs, such as \texttt{fakedaum}, are more challenging than others.
Although the overapproximation times out on \texttt{scipiex}, our MCFL reachability computes it within minutes. 
We also note that the whole benchmark set contains 3 larger graphs (namely, \texttt{phospy}, \texttt{simhosy}, and \texttt{skullkey}), comprising tens of thousands of nodes, for which both $\Language(\Grammar_2^+)$ reachability and the overapproximation time out.
Nevertheless, we are able to confirm that also in these graphs, $\Language(\Grammar_2^+)$ reachability reports significantly more reachable pairs than $\Language(\Grammar_1^+)$, highlighting the importance of context-sensitivity.
Finally, we remark that we are also able to run $\Language(\Grammar_3^+)$ reachability on some of the smaller benchmarks. Unfortunately, none of the benchmarks in which the analysis is not complete
(\texttt{batterydoc}, \texttt{droidkonfu}, and \texttt{fakedaum}) ran within the time limit, so we could not report any additional reachable pairs from this analysis.
This is in alignment with our observations on the coverage of $\Language(\Grammar_d^+)$ according to \cref{table:results}.

\Paragraph{Witness reporting.} 
We remark one additional benefit of MCFL reachability.
The MCFL derivations witnessing reachability naturally (i.e., without a complexity penalty) produce reachability witnesses, in the form of paths, that can be reported to the programmers for further assistance.
For example, in the case of taint analysis, the programmers can precisely recover the tainted paths reported by $\Language(\Grammar_2^+)$ in order to debug their code and remove the data leak.
In contrast, overapproximate methods do not, in general, produce concrete witnesses.

\Paragraph{Conclusion.}
In conclusion, our experimental results indicate that MCFL reachability is a useful underapproximation of interleaved Dyck reachability, and when combined with overapproximate methods, produces results that are often complete, and typically of very high coverage.
This coverage can further be reached with mild context sensitivity, using MCFLs of only $2$ dimensions.
Finally, the algorithm often produces results in reasonable running time.
For the more challenging inputs, further optimizing heuristics can have a significant impact, as has been the case so far for standard CFL reachability.
We leave this direction open for interesting future work.
\section{Related Work}\label{SEC:RELATED_WORK}

Context-sensitive language models are a key step towards static analyses of higher precision, and have been utilized in diverse ways.
For example,~\cite{Tang2017} uses TAL reachability to capture conditional interprocedural reachability on library code, in the absence of client code.
Interleaved Dyck reachability, in particular, has been also approximated using linear conjunctive languages (LCLs)~\cite{Zhang2017},
and mixed integer programming (MIP)~\cite{Li2023}.
Unfortunately, LCL reachability is undecidable, which required resorting to a fast but approximate algorithm, while MIP has exponential complexity and only reasons about the Parikh images of the corresponding languages.
The hardness of the problem stems from the fact that the underlying pushdown automata is operating on multiple stacks.

At its core, the problem of interleaved Dyck reachability also arises in distributed models of computation, represented as a pushdown automaton with multiple stacks.
Naturally, the reachability problem for even 2-stack pushdown automata is undecidable, thus various approximations have been considered in the literature.
One such approach is by bounding the scope, keeping track of the last $k$ symbols pushed onto each stack~\cite{laTorre20}, thereby approximating the interleaved language by a regular one.
Another approach is bounded-context switching~\cite{Qadeer2005, Bansal2013,Chatterjee2017}, which bounds the number of times the automaton can switch between different stacks. 
This concept is more related to our MCFG construction, as the grammar presented in \cref{SUBSEC:BASIC_RULES} essentially emulates $d$-context switching.
However, the stronger grammar in \cref{SUBSEC:EXPRESSIVENESS} goes beyond $d$ contexts.
A further generalization of bounded context-switching is bounded phases~\cite{Torre2007}.
This approximation is incomparable to the MCFG of \cref{SUBSEC:EXPRESSIVENESS}, in the sense that the MCFG for bounded dimensions recognizes strings that need an unbounded number of phases, and strings recognized with a bounded number of phases require an MCFG of arbitrary dimension.
It is worth noting, though, that the reachability problem for bounded-phase multi-stack
pushdown automata has a double exponential dependency on the number of phases $k$~\cite{Torre2007}, while MCFL-reachability only has a single exponential dependency on the number of dimensions $d$.
Analogous observations hold for the class of ordered multi-stack automata~\cite{Breveglieri1996}.
The reachability problem for many such restricted classes of multi-stack automata has also been studied from the perspective of treewidth~\cite{Madhusudan2011}. 

Besides interleaving Dyck languages, various static analyses have been phrased by interleaving other types of CFLs~\cite{UA6, UA7, UA8}.
Our MCFL-based approximations for interleaved Dyck can be extended naturally to interleaving arbitrary CFLs, though the precise approximation that works best in each case merits independent investigation.

Due to the importance of CFL reachability in static analysis, the complexity of the problem has been studied extensively under the fine-grained lens.
The cubic bound $O(n^3)$ is believed to be tight as the problem is 2NPDA-hard~\cite{Heintze97}.
Other fine-grained lower bounds have also been proved, highlighting the difficulty of the problem~\cite{Koutris2023,Chatterjee2018,Chistikov2022}.
For single-letter alphabets, the problem takes essentially matrix-multiplication time~\cite{Mathiasen2021}, and this bound is optimal~\cite{Hansen2021}.
See~\cite{Pavlogiannis2022} for a recent survey.
The program model is also commonly known as Recursive State Machines~\cite{Alur2005,Chaudhuri2008,Chatterjee2019}.
In this work we have extended this line of research by 
(i)~giving generic upper bounds for MCFL reachability, from which the cubic bound of CFLs follows as a special case, and
(ii)~proving fine-grained lower bounds showing that the complexity increase along the MCFL hierarchy is unavoidable.

\section{Conclusion}\label{SEC:CONCLUSION}
In this paper, we have introduced MCFL reachability as a robust, expressive, and tunable context-sensitive model for static analyses, generalizing the famous CFL reachability formalism.
We have demonstrated its potential by developing MCFLs to approximate interleaved Dyck reachability, which is abundant in static analyses.
Our results show that sound and high coverage results can be obtained already with 2-dimensional MCFLs.
Moreover, we have developed a generic algorithm for MCFL reachability, and proven lower bounds which are essentially tight for MCFLs of rank 1.

MCFL reachability raises several interesting open questions to further impact static analyses.
First, we expect that the power of MCFLs will be utilized in several other static analysis settings, such as, 
(i)~for the more precise modeling of higher order programs, and 
(ii)~for accurate call-graph construction (current CFL models typically require a call graph be given in advance, which is often imprecise).
This is so because, in such cases, the way program execution develops within a function depends on the calling context, thereby requiring a context-sensitive language to model it precisely.
Second, following recent developments on algorithmic aspects of CFL reachability,
MCFL reachability must be better understood, both in terms of further fine-grained complexity lower bounds, as well as practical heuristics.

\section{Acknowledgements}\label{SEC:ACKNOWLEDGMENTS}

We thank anonymous reviewers for their constructive feedback. 
G.K. Conrado was supported by the Hong Kong PhD Fellowship Scheme (HKPFS).
A. Pavlogiannis was partially supported by a research grant (VIL42117) from VILLUM FONDEN.

\pagebreak

\bibliographystyle{ACM-Reference-Format}
\bibliography{bibliography}

\newpage

\appendix

\section{Proofs of \cref{SEC:MODELING}}\label{SEC:APP_MODELING}

\lemsubslemmaseparation*
\begin{proof}
We will prove this by induction on the size of $s$, i.e, let $s'=x'_1x'_2\cdots x'_{c'}\in \Dyck_\alpha $ be such that either $|s'|<|s|$, or $c'<c$, then $P^{c'}(x'_1,\cdots,x'_{c'})$. The basis is when $s=e$, in which case $s$ is recognized by $P^{c}$ for every ${c}$, by the terminal rule \ref{subs-rule-k-epsilon}, or $c=1$, in which case $P^1$ recognizes $s$ as its rules are equivalent to the rules of $P$ in $\Dyck_\alpha$. Firstly we consider the following two cases.

\begin{itemize}
\item If $x_1=\epsilon$, then $P^{c}(x_1,x_2,\cdots,x_{c})$ can be derived, by rule \ref{subs-rule-k-concatenation} from $P^2(\epsilon,\epsilon)$ and $P^{c-1}(x_2,\cdots,x_c)$.
\item If $x_c=\epsilon$, then $P^k(x_1,x_2,\cdots,x_c)$ can be derived, by rule \ref{subs-rule-k-concatenation} from $P^{c-1}(x_1,\cdots,x_{c-1})$ and $P^2(\epsilon,\epsilon)$.
\end{itemize}

If none of the cases above apply, then consider the last derivation rule of $P(s)$ in $\Dyck_\alpha$. If that is Rule~\ref{rule-parentheses}, and $x_1,x_c \neq \epsilon$, then $s$ can be written as $\pcOne{\op_i}x'_1x_2\cdots x_{c-1}x'_c\pcOne{\cp_i}$ where $x_1=\pcOne{\op_i}x'_1$ and $x_c=x'_c\pcOne{\cp_i}$. By the induction hypothesis, $P^c(x'_1,x_2,\cdots,x_{c-1},x'_c)$. By rule \ref{subs-rule-k-parentheses}, $P^c(\pcOne{\op_i}x'_1,x_2,\cdots,x_{c-1},x'_c\pcOne{\cp_i})$ and thus $P^c(x_1,x_2,\cdots,x_{c-1},x_c)$. 

If that is rule \ref{rule-concatenation}, then $s=s's''$ where $s',s''\neq \epsilon$ and $s',s''\in \Dyck_\alpha$. Let $a$ be the largest number such that $x_a$ contains a character in $s'$. By the induction hypothesis, $P^a(x_1, \cdots, x_{a-1}, x'_a)$, where $x'_a$ is the prefix of $x_a$ formed by the characters in $s'$. Let $x_a=x'_ax''_a$ (notice that it is possible that $x''_a=\epsilon$). By the induction hypothesis,  $P^{c-a+1}(x''_a, x_{a+1}, \cdots, x_c)$. Therefore, by rule \ref{subs-rule-k-concatenation}, $P^c(x_1, \cdots,x_{a-1} , x'_a x''_a, x_{a+1}, \cdots, x_c)$, and thus $P^c(x_1,\cdots,x_a,\cdots, x_c)$.
\end{proof}

\lemsubslemmaseparationt*
\begin{proof}
The proof is analogous to \cref{lem:subs-lemma-separation}.
\end{proof}

\lemmathereisdimension*
\begin{proof}
Since $s \in \IDL$, all characters in $s$ are either in $\{\pcOne{\op_i},\pcOne{\cp_i}\}_{i \in [k]}$ or $\{ \pcTwo{\ob_i},\pcTwo{\cb_i}\}_{i \in [k]}$, thus $s$ can be written as
$x_1 y_1 \cdots x_{d_0} y_{d_0}$ where $x_i$ is composed of characters in $ \{\pcOne{\op_i},\pcOne{\cp_i}\}_{i \in [k]}$ and $y_i$ is composed characters in  $\{ \pcTwo{\ob_i},\pcTwo{\cb_i}\}_{i \in [l]}$. Notice that some $x_i$ or $y_i$ may be $\epsilon$. 

Notice that from $s \in \IDL$, $(s\Project \{\pcOne{\op_i},\pcOne{\cp_i}\}_{i \in [k]})\in \Dyck_\alpha$, thus from \cref{lem:subs-lemma-separation}, it must be the case that $P^{d_0}(x_1,\cdots,x_{d_0})$, and, from \cref{lem:subs-lemma-separation-t}, it must be the case that $Q^{d_0}(y_1,\cdots,y_{d_0})$. Thus, by Rule~\ref{subs-rule-k-interleave}, $S$, and by consequence, $\Grammar_{d_0}^\circ$, recognizes $s$.

Now consider some $d>d_0$. We can instantiate $P^{d_0-d}(\epsilon,\cdots, \epsilon)$ by Rule~\ref{subs-rule-k-epsilon}. Then, by concatenating $P^{d_0}(x_1,\cdots,x_{d_0})$ and $P^{d_0-d}(\epsilon,\cdots, \epsilon)$ through Rule~\ref{subs-rule-k-concatenation}, we produce $P_{d}(x_1,\cdots, x_{d_0}, \epsilon,\cdots, \epsilon)$. Similarly, we may generate $Q_{d}(y_1,\cdots, y_{d_0}, \epsilon,\cdots, \epsilon)$. Thus, by Rule~\ref{subs-rule-k-interleave}, $S$, and by consequence, $\Grammar_{d}^\circ$, recognizes $s$.
\end{proof}

\leminsertstringdyck*
\begin{proof}
We will prove this by induction on the number of derivation steps of $S(s)$. 

Firstly if $|s_1|=0$ or $|s_2|=0$ then $u$ can be derived by Rule~\ref{rule-concatenation}, through $S(ts_2) \gets S(t),S(s_2)$ or $S(s_1t) \gets S(s_1),S(t)$, respectively. Otherwise, let us consider the last derivation step of $s$. This cannot be Rule~\ref{rule-epsilon}, otherwise $|s_1|=|s_2|=0$ and this case would have been handled. 

If the last derivation is through Rule~\ref{rule-parentheses}, then, since $|s_1|>0$ and $|s_2|>0$, $s=\pcOne{\op_i}s_1's_2'\pcOne{\cp_i}$ for some $i$. Since $s$ was derived by $S(\pcOne{\op_i} s_1's_2' \pcOne{\cp_i}) \gets S(s_1's_2')$, then $s_1's_2'\in \Language(\Dyck_\alpha)$, thus, by the induction hypothesis, $s_1'ts_2'\in \Language(\Dyck_\alpha)$ and so by applying $S(\pcOne{\op_i} s_1'ts_2' \pcOne{\cp_i}) \gets S(s_1'ts_2')$, $\pcOne{\op_i} s_1'ts_2' \pcOne{\cp_i} = u\in \Language(\Dyck_\alpha)$.

If the last derivation is through Rule~\ref{rule-concatenation}, then $s=s_1's_2'$ for some $s_1',s_2'\in \Language(\Dyck_\alpha)$. Suppose $|s_1|<|s_1'|$, then $s_1'=s_1v$ for some $v$. By the induction hypothesis, since $s_1v\in \Language(\Dyck_\alpha)$ then $s_1tv\in \Language(\Dyck_\alpha)$, thus by applying $S(s_1tvs_2') \gets S(s_1tv),S(s_2')$ we have $s_1tvs_2' = s_1ts_2 = u\in \Language(\Dyck_\alpha)$. The case when $|s_1|>|s_1'|$ is analogous.

\end{proof}

\leminterdyckmcfginvariants*
\begin{proof}

Consider some string $s$ such that $S(s)$. Its final derivation must have been through Rule~\ref{subs-rule-k-interleave} or \ref{subs-rule-k-interleave-2}, so $s=x_1 y_1 \cdots x_d y_d$ or $s=y_1 x_1 \cdots y_d x_d$ for some $x_i,i\in[d]$ and $y_i,i\in[d]$ such that $S^d(x_1, \cdots, x_d)$ and $ Q^d(y_1, \cdots, y_d)$. We will assume without loss of generality that $s=x_1 y_1 \cdots x_d y_d$.

By its definition, $s\in \IDL$ if $(s\Project \{\pcOne{\op_i},\pcOne{\cp_i}\}_{i \in [n]}) \in \Dyck_\alpha$ and $(s\Project \{\pcTwo{\ob_i},\pcTwo{\cb_i}\}_{i \in [m]}) \in \Dyck_\beta$. Let us firstly consider $(s\Project \{\pcOne{\op_i},\pcOne{\cp_i}\}_{i \in [n]})$. By Invariant~\ref{eq:rules_invariant_parentheses}, since $S^d(x_1,\cdots x_d)$, $x_1\cdots x_d\Project \{\pcOne{\op_i},\pcOne{\cp_i}\}_{i \in [k]}\in\Dyck_\alpha$. Since $Q^d(y_1,\cdots, y_d)$, for all $y_j$, $y_j \Project \{\pcOne{\op_i},\pcOne{\cp_i}\}_{i \in [n]}\in\Dyck_\alpha$. Thus by repeated applications of \cref{lem:insert_string_dyck} we may insert the elements $y_j$ in $x_1\cdots x_d$ and obtain that $x_1 y_1 \cdots x_d y_d\Project \{\pcOne{\op_i},\pcOne{\cp_i}\}_{i \in [n]}\in\Dyck_\alpha$.

The proof that $x_1 y_1 \cdots x_d y_d\Project \{\pcTwo{\ob_i},\pcTwo{\cb_i}\}_{i \in [m]}\in\Dyck_\beta$ is analogous.

\end{proof}

\section{Proofs of \cref{SEC:MCFG_REACHABILITY}}\label{SEC:APP_MCFG_REACHABILITY}

\lemnormalform*
\begin{proof}
We transform $\Grammar$ to $\Grammar'$ by applying the following steps in sequence.
\begin{compactenum}
\item\label{step1}
Firstly we ensure that all basic rules have arity 1. 
While there exists a basic rule of the form $A(w_1,\ldots,w_k)$ for $k>1$ and each $w_i\in\Terminals^*$, we replace it with two rules
\begin{align*}
&A'(w_1)\\
&A(x,w_2,\dots,w_k)\gets A'(x)
\end{align*}
\item\label{step2} 
We then ensure that all basic rules are of Type~1. While there exists a basic rule of the form $A(w)$ with $w=a w'$ for $a\in \Terminals$ and $w'\in \Terminals^+$, we replace it with two rules
\begin{align*}
&A'(a)\\
&A(xw')\gets A'(x)
\end{align*}
\item\label{step3} 
Now we ensure that every rule with more than one non-terminal on the right-hand side does not consume any terminals and hence is of Type~5.

While there is a rule $A_0(s_1,\dots, s_{k_0}) \gets A_1(x_1^1,\dots, x_{k_1}^1) \cdots A_{\ell}(x_1^{\ell},\dots, x_{k_{\ell}}^{\ell})$ with $\ell\geq 2$ and such that there exists some $s_i$ that contains a terminal, we distinguish three cases:
If $s_i\in \Terminals^+$, we replace the rule with
\begin{align*}
&A'(s_1,\dots, s_{i-1},  s_{i+1},\dots, s_{k_{0}}) \gets A_1(x_1^1,\dots, x_{k_1}^1) \cdots A_{\ell}(x_1^{\ell},\dots, x_{k_{\ell}}^{\ell})\\
&A_0(x_1,\dots, x_{i-1}, s_i, x_{i+1},\dots x_{k_{0}})\gets A'(x_1,\dots, x_{i-1}, x_{i+1},\dots,x_{k_0})
\end{align*}

Else, if $s_i$ has the form $s_i=wx^{j}_{m}s'_i$, for $w\in \Terminals^+$, we replace the rule with
\begin{align*}
&A'(x_1,\dots, x_{m-1}, wx_{m}, x_{m+1},\dots, x_{k_{j}}) \gets A_{j}(x_1,\dots, x_{k_{j}})\\
&A_0(s_1,\dots, s_{i-1}, x^{j}_{m}s'_i, s_{i+1},\dots, s_{k_0}) \gets A_1(x_1^1,\dots, x_{k_1}^1) \cdots A_{\ell}(x_1^{\ell},\dots, x_{k_{\ell}}^{\ell})
\end{align*}
where $A_j$ on the right-hand side of the second rule is replaced by $A'$.

Otherwise, $s_i$ has the form $s_i=s^1_ix^{j}_{m}ws^2_i$, for $w\in \Terminals^+$, and we replace the rule with
\begin{align*}
&A'(x_1,\dots, x_{m-1}, x_{m}w, x_{m+1},\dots, x_{k_{j}}) \gets A_{j}(x_1,\dots, x_{k_{j}})\\
&A_0(s_1,\dots, s_{i-1}, s^1_ix^{j}_{m}s^2_i, s_{i+1},\dots, s_{k_0}) \gets A_1(x_1^1,\dots, x_{k_1}^1) \cdots A_{\ell}(x_1^{\ell},\dots, x_{k_{\ell}}^{\ell})
\end{align*}
where $A_j$ on the right-hand side of the second rule is replaced by $A'$. 
\item\label{step4} 
We now ensure that whenever there is a rule of the form 
\[
A(s_1,\dots,s_k)\gets B(x_1,\dots, x_{k'})
\]
and there exists $s_i$ and $s_j$ that are both are not a single variable, each of them contains at least one variable.

While there exists a rule  of the form
\[
A(s_1,\dots,s_k)\gets B(x_1,\dots, x_{k'})
\]
and there exist $i,j\in[k]$ such that $s_i\in \Terminals^*$ and $s_j$ is not a single variable, we replace the rule with
\begin{align*}
&A'(s_1,\dots, s_{i-1}, s_{i+1}, \dots s_k) \gets B(x_1,\dots, x_{k'})\\
&A(x_1, \dots, x_{i-1}, s_i, x_{i},\dots, x_{k-1})\gets A'(x_1, \dots, x_{k-1})
\end{align*} 
\item\label{step5} 
We now ensure that every production rule that reduces the arity on the left hand side is of Type~5. While there is a rule of the form
\[
A(s_1,\dots,s_k)\gets B(x_1,\dots, x_{k'})
\]
and there exists some $i\in[k]$ such that $s_i$ contains at least two variables and $|s_i|>2$,  
we decompose $s_i$ into two substrings $s_i=s_i^1s_i^2$ such that each of them contains a variable, and replace the rule with
\begin{align*}
&A'(s_1,\dots,s_{i-1}, s_i^1,s_i^2,s_{i+1},\dots,s_k)\gets B(x_1,\dots, x_{k'})\\
&A(x_1,\dots, x_{i-1}, x_ix_{i+1}, x_{i+2},\dots, x_{k+1}) \gets A'(x_1,\dots, x_{k+1}).
\end{align*}
Note that, because of Step~\ref{step4}, we have $k+1 \leq k'$, and hence the dimension of the grammar does not increase. 
\item\label{step6} 

We now ensure that all arguments in the left-hand-side of all rank-1 rules
are of the proper shape $x$, $a$, $ax$, $xa$ or $xy$. While there is a rule of the form
\[
A(s_1,\dots,s_k)\gets B(x_1,\dots, x_{k'})
\]
and there exists some $i\in[k]$ such that $|s_i|>2$ or $s_i$ consists of two terminals, we decompose $s_i$ into two substrings $s_i=s_i^1s_i^2$.
If $s_i^1$ contains a variable, note that $s_i^2\subseteq \Terminals^+$ due to Step~\ref{step5}.
We replace the rule with
\begin{align*}
&A'(s_1,\dots,s_{i-1}, s_i^1,s_{i+1},\dots,s_k)\gets B(x_1,\dots, x_{k'})\\
&A(x_1,\dots, x_{i-1}, x_is_i^2, x_{i+1},\dots, x_{k}) \gets A'(x_1,\dots, x_k).
\end{align*}
Otherwise, $s_i^1\in\Terminals^+$ and we replace the rule with
\begin{align*}
&A'(s_1,\dots,s_{i-1}, s_i^2,s_{i+1},\dots,s_k)\gets B(x_1,\dots, x_{k'})\\
&A(x_1,\dots, x_{i-1}, s_i^1x_i, x_{i+1},\dots, x_{k}) \gets A'(x_1,\dots, x_k).
\end{align*}
Note that if $s_i^2\in\Terminals^+$ then $s_i\in\Terminals^+$, so this step
cannot introduce a new instance of Step~\ref{step3}.
%

\item\label{step7} 
Finally, we ensure that all production rules of rank 1
are of Type~2, Type~3, or Type~4. While there exists a rule  of the form
\[
A(s_1,\dots,s_k)\gets B(x_1,\dots, x_{k'})
\]
and there exist $i,j\in[k]$ such that both $s_i$ and $s_j$ are not a single variable, let the sequence of variables in $s_i$ be
$y_1,\ldots,y_p$. We replace the rule with
\begin{align*}
&A'(s_1,\dots, s_{i-1}, y_1,\dots, y_p, s_{i+1}, \dots, s_k) \gets B(x_1,\dots, x_{k'})\\
&A(x_1, \dots, x_{i-1}, s_i, x_{i+1},\dots, x_{k})\gets 
 A'(x_1, \dots, x_{i-1}, y_1, \dots, y_p, x_{i+1},\dots,x_k)
\end{align*}
Note that, because of Step~\ref{step4}, each $s_j$ contains a variable, so $k+p \leq k'$, and hence the dimension of the grammar does not increase.
\end{compactenum}
Finally, observe that $\Grammar$ and $\Grammar'$ have the same dimension and rank, while $|\Grammar'|=O(\Poly(|\Grammar|))$.
The construction can easily be carried out in $O(\Poly(|\Grammar|))$, though details are omitted for brevity.
\end{proof}

\lemsoundness*
\begin{proof}
The statement follows by a straightforward induction.
It clearly holds in the initialization loop (\cref{line:algo_type1}).
Moreover, assuming that it holds for  elements $A_i[(u^i_1, v^i_1),\dots, (u^i_{k_i},v^i_{k_i})]$ when extracted from $\Worklist$ (\cref{line:algo_extraction}), it also holds for any element $A[(u_1,v_1),\dots(u_{k},v_{k})]$ inserted in $\Worklist$ in that iteration,
as the algorithm closely follows the semantics of derivations $\Derives$.
As the details are straightforward, they are omitted for brevity.
\end{proof}

\lemcompleteness*
\begin{proof}
The proof is by induction on the structure of derivations. We say that a pair of nodes $(u,v)$ \emph{expands into} $(u',v')$ if there are paths $\{P_1\colon u'\Path u\}$ and $\{P_2\colon v\Path v'\}$ such that $\Label(P_1)=\Label(P_2)=\epsilon$. Analogously, element $A[(u_1, v_1),\dots, (u_k,v_k)]$ expands into $A[(u'_1, v'_1),\dots, (u'_k,v'_k)]$ if for every $i\in [k], (u_i,v_i)$ expands into $(u'_i, v'_i)$.

Firstly consider the rules $A$ of Type~1. When that rule is of type $A(\epsilon)$, \cref{line:epsilon_edges} guarantees that every empty path, $\{P\colon u\Path u\}$ will be added as a derivation of $A$ in the initialization in \cref{line:algo_type1}. Furthermore, \cref{line:for_expand_epsilon} guarantees that if element $A[(u, v)]$ is inserted in $W$ then every element it expands into will also be inserted in $W$. This means that for every path $P\colon u \Path v$ such that $\Label(P)=\epsilon$, the element $A[(u,v)]$ will be inserted in $W$.

For the other rules of Type~1, i.e, the rules $A(a)$ for some $a\in \Sigma$, for any edge $(u,v)$ of label $a$ the element $A[(u,v)]$ will be inserted in $W$ by the initialization in \cref{line:algo_type1}. \cref{line:for_expand_epsilon} then guarantees that every element $A[(u,v)]$ expands into will also be inserted into $W$. This means that for every path $P\colon u \Path v$ such that $\Label(P)=a$, the element $A[(u,v)]$ will be inserted in $W$.

Thus the statement holds for all rules $A$ of Type~1.

Now consider an arbitrary sequence of paths $\{P_i\colon u_i\Path v_i\}_{i\in [k]}$ such that $A\Derives(\Label(P_1), \dots, \Label(P_k))$. 
Following the structure of derivations, there is a rule
\[
A(s_1,\dots, s_{k_0}) \gets A_1(x_1^1,\dots, x_{k_1}^1) \cdots A_{\ell}(x_1^{\ell},\dots, x_{k_{\ell}}^{\ell})
\numberthis\label{eq:production_rul2}
\]
and paths $\{P^i_{j}\}_{i\in [\ell], j\in[k_i]}$ such that $A_i\Derives(\Label(P^i_1),\dots, \Label(P^i_{k_i}))$ for every $i\in[\ell]$. Notice that for every element $A_i[(u_1, v_1),\dots, (u_{k_i},v_{k_i})]$ in $W$, \cref{line:for_expand_epsilon} will insert in $W$ the derivations $A_i[(u_1, v_1),\dots, (u_{k_i},v_{k_i})]$ expands into, thus we must only inductively handle the minimal form of a derivation.
We distinguish cases based on the type of the rule according to the normal form of \cref{SUBSEC:NORMAL_FORM}.
\begin{compactitem}
\item[\emph{Type~2}:] 
We have $\ell=1$, and \cref{eq:production_rule} has the form $A(x_1,\dots,x_{i-1},ax_i,x_{i+1},\dots,x_k)\gets B(x_1,\dots,x_k)$, for some $a\in \Terminals$, and $B\Derives(\Label(P_1),\dots, \Label(P_k))$, for paths $P_1,\dots,P_k$.
By the induction hypothesis, the algorithm inserts the element $B[(u'_1, v'_1),\dots, (u'_k, v'_k)]$ in $\Worklist$, where $u'_i, v'_i$ are the endpoints of the corresponding paths.
In the iteration where $B[(u'_1, v'_1),\dots, (u'_k, v'_k)]$ is extracted from $\Worklist$, the algorithm (\cref{line:algo_type2}) generates
\begin{align*}
A[(u'_1, v'_1),\dots, (u'_{i-1}, v'_{i-1}), (u',v'_i), (u'_{i+1}, v'_{i+1}),  (u'_k,v'_k)] =\\
= A[(u_1, v_1),\dots, (u_k,v_k)]
\end{align*}
and inserts it in $\Worklist$.
\item[\emph{Type~3}:] This case is symmetric to Type~2.
\item[\emph{Type~4}:]
We have $\ell=1$, and \cref{eq:production_rule} has the form $A(x_1,\dots, x_{i-1}, a, x_{i},\dots, x_{k})\gets B(x_1,\dots,x_k)$, for some $a\in \{\Terminals\cup \epsilon\}$, and $B\Derives(\Label(P_1),\dots, \Label(P_k))$, for paths $P_1,\dots,P_k$.
By the induction hypothesis, the algorithm inserts the element $B[(u'_1, v'_1),\dots, (u'_{k-1}, v'_{k-1})]$ in $\Worklist$, where $u'_i, v'_i$ are the endpoints of the corresponding paths.
In the iteration where $B[(u'_1, v'_1),\dots, (u'_{k-1}, v'_{k-1})]$ is extracted from $\Worklist$, the algorithm (\cref{line:algo_type4}) generates
\begin{align*}
A[(u'_1, v'_1),\dots, (u'_{i-1}, v'_{i-1}), (u,v), (u'_{i+1}, v'_{i+1}),  (u'_{k-1},v'_{k-1})] =\\
= A[(u_1, v_1),\dots, (u_k,v_k)]
\end{align*}
and inserts it in $\Worklist$.
\item[\emph{Type~5}:]
By the induction hypothesis, we have that each $A_i([(u^i_1,v^i_1),\dots, (u^i_{k_i}, v^i_{k_i})])$ is inserted in $\Worklist$,
where $u^i_j, v^i_j$ are the endpoints of the path $P^i_j$.
Consider the iteration where the last such $A_i$ is extracted from $\Worklist$.
We have that $s_1\cdots s_{k_0}\PathDerives \{ (u^i_{j_i}, v^i_{j_i}) \}_{i\in[\ell],j_i\in[k_{i}]}$, and
the algorithm (\cref{line:algo_type5}) generates 
$A[(u_1,v_1),\dots, (u_{k_0},v_{k_0})]$ and inserts it in $\Worklist$,
as we must have that $(u_t,v_t)_t=\Endpoints(s_t,\{ (u^i_{j_i}, v^i_{j_i}) \}_{i\in[\ell],j_i\in[k_{i}]})$ for each $t\in [k_0]$.

\end{compactitem}
\end{proof}

\lemcomplexity*
\begin{proof}
Observe that for every non-terminal $A$ with arity $k$, for every sequence of nodes $(u_1,v_1),\dots, (u_k,v_k)$,  the element $A[(u_1,v_1),\dots, (u_k,v_k)]$ is inserted in $\Worklist$ at most once.
Hence an element with non-terminal $A$ is inserted in $\Worklist$ at most $n^{2k}=O(n^{2d})$ times in total.
When this element is extracted from $\Worklist$, the algorithm attempts to use it as the right hand side of a production rule.
The algorithm spends $O(\Poly(|\Grammar|)\cdot \delta)$ time for expanding the predicate through $\epsilon$-edges (\cref{line:for_expand_epsilon}), $O(\Poly(|\Grammar|)\cdot \delta)$ time for matching this element on a rule of Type~2 (\cref{line:algo_type2}) and Type~3 (\cref{line:algo_type3}), and it spends $O(\Poly(|\Grammar|)\cdot n^2)$ time for matching it on a rule of Type~4 (\cref{line:algo_type4}).
Note, however, that for Type~4 rules we have $k\leq d-1$, as a Type~4 rule increases the arity on the left hand side to an arity $\leq d$.
Thus, without considering Type~5 rules, the total running time is asymptotically bounded by $\Poly(|\Grammar|)$ multiplied by
\begin{align}
\left(\delta\cdot n^{2d}  + n^{2(d-1)} \cdot n^2\right) = O(\delta \cdot n^{2d} )
\label{eq:type1to4_bound}
\end{align}

We now turn our attention to the time spent processing Type~5 rules (\cref{line:algo_type5}).
For each such rule, \cref{line:algo_type5_pathmatching} is executed at most once for each set of nodes $\{ (u^i_{j_i}, v^i_{j_i}) \}_{i\in[\ell],j_i\in[k_{i}]}$ such that
$s_1\cdots s_{k_0}\PathDerives \{ (u^i_{j_i}, v^i_{j_i}) \}_{i\in[\ell],j_i\in[k_{i}]}$.
For each such $s_i$, the number of subsets $X$ of $\{ (u^i_{j_i}, v^i_{j_i}) \}_{i\in[\ell],j_i\in[k_{i}]}$ such that $s_i\PathDerives X$ is at most $n^{|s_i|+1}$.
Hence the total time spent in Type~5 rules is asymptotically bounded by $\Poly(|\Grammar|)$ multiplied by
\begin{align}
\prod_{i\in[\ell]} n^{|s_i|+1} = n^{\sum_{i\in[\ell]}(|s_i|+1)} = n^{\ell + \sum_{i\in[\ell]}|s_i|} \leq n^{d(r+1)}
\label{eq:type5_bound}
\end{align}
as $\ell\leq r$ and $\sum_{i\in[\ell]}|s_i|\leq \sum_{j\in[r]} k_{j} \leq r\cdot d$.
Observe that when $r=1$, the bound of \cref{eq:type5_bound} yields $n^{2d}$ and is thus dominated by the bound in \cref{eq:type1to4_bound}, yielding \cref{item:complexity_r1} of the lemma.
For $r>1$, \cref{eq:type5_bound} is at least $n^{3d}$, and thus Type~5 rules dominate the complexity, yielding \cref{item:complexity_rlarge} of the lemma.
\end{proof}

\section{Proofs of \cref{SEC:LOWER_BOUNDS}}\label{SEC:APP_LOWER_BOUNDS}

\thmovhard*
\begin{proof}
Let $X_1,\dots, X_{\NumVectorSets}$ be an instance of $\NumVectorSets$-OV, and let $w\in \TemplateLanguage_{\NumVectorSets, \NumVectors}^{\NumBits}$ be the string-encoding of this instance.

We start with the correctness of the reduction.
It is straightforward to see that $\StartSymbol\Derives(w)$ iff there exist $j_1,\dots, j_{\NumVectorSets}$ such that
\[
B\Derives( w^1_{j_1} \#_1x_1\#_2 w^2_{j_2}, \dots,  w^{\NumVectorSets-1}_{j_{\NumVectorSets-1}} \#_{\NumVectorSets-1}x_{d}\#_{\NumVectorSets} w^{{\NumVectorSets}}_{j_{\NumVectorSets}})
\]
where $x_1, \dots, x_d$ are substrings in our template language $\TemplateLanguage_{\NumVectorSets, \NumVectors}^{\NumBits}$.
In turn, this happens iff for every $\ell\in [\NumBits]$, there is some $i\in [\NumVectorSets]$ such that
(i)~$w^{i}_{j_i}[\ell]=0$, if $i$ is even, or
(ii)~$w^{i}_{j_i}[\NumBits-\ell+1]=0$, if $i$ is odd.
Finally, each of (i) and (ii) happens iff the corresponding vector has $x^{i}_{j_i}[\ell]=0$.
Thus, $\StartSymbol\Derives(w)$ iff there exist $j_1,\dots, j_{\NumVectorSets}$ such that the vectors $x^1_{j_1},\dots x^{\NumVectorSets}_{j_{\NumVectorSets}}$ are orthogonal.

We now establish the complexity.
The string $w$ has length $|w|=O(\NumVectors\cdot \NumVectorSets\cdot \NumBits)$, while $|\Grammar_{\NumVectorSets}|$ has constant size that depends only on $\NumVectorSets$.
Now let $d=\NumVectorSets/2$, and assume that there exists some fixed $\epsilon>0$ such that membership in $\Language(\Grammar_{\NumVectorSets})$ is decided in $O(|w|^{2d-\epsilon})$ time.
Then we can solve $2\NumVectorSets$-OV in time
\[
O((\NumVectors\cdot \NumVectorSets\cdot \NumBits)^{2d-\epsilon})=O(\NumVectors^{2d-\epsilon}\cdot \Poly(\NumBits)) = O(\NumVectors^{\NumVectorSets-\epsilon}\cdot \Poly(\NumBits))
\]  
time, which contradicts OVH.
Thus we have a fine-grained reduction with respect to the time bounds of \cref{thm:ov_hard}, and the theorem follows.
\end{proof}

\thmtrianglehard*
\begin{proof}[Proof of \cref{thm:triangle_hard}]
We first argue about the correctness of the reduction, i,.e., $v$ is $\TriangleL$-reachable from $u$ in $G$ iff $G'$ contains a triangle.
Indeed, observe that for any $k\in[n]$, the graph $G$ has a $\TriangleL$-labeled path $P_k\colon u_k\Path z_k$ iff there is a $j\in[n]$ such that $j$ is both a distance-one and distance-two neighbor of $i$ in $G'$.
In turn, this is equivalent to saying that $i$ and $j$ form a triangle with a third node in $G'$.
Next, observe that we have a single path $P^u_k\colon u\Path u_k$ for every $k\in[n]$, with $\Label(P^u_k)=0^k$.
Similarly, we have a single path $P^v_k\colon z_k\Path  v$ for every $k\in [n]$, with $\Label(P^v_k)=1^k$.
Thus, we have a path $P\colon u\Path v$ iff there exists some $k\in [n]$ such that $P=P^u_k\circ P_k \circ P^v_k$, and $\Label(P)=0^k1^k\in \TriangleL$.

Regarding the complexity, we clearly have $|\Grammar|=O(1)$,
while the size of $G$ is proportional to the size of $G'$, and thus we have a fine-grained reduction.
\end{proof}

\end{document}